\documentclass[
  aps,
  prx,
  onecolumn,         % Single-column layout
  floatfix           % Better float placement
]{revtex4-2}

\usepackage[utf8]{inputenc}
\usepackage[T1]{fontenc}
\usepackage{lmodern}
\usepackage[english]{babel}

\usepackage{amsthm,amsmath,amssymb,amsfonts,mathrsfs,bbold,bm}
\usepackage{dsfont}
\usepackage{newtxtext,newtxmath}

\usepackage{graphicx}
\usepackage{float}
\usepackage{tikz}
\usetikzlibrary{calc,decorations.markings,arrows.meta,positioning}

\usepackage{array,longtable,tabularx}

\usepackage{algorithm}
\usepackage{algpseudocode}
\algdef{SE}[SUBALG]{Indent}{EndIndent}{}{\algorithmicend\ }%
\algtext*{Indent}\algtext*{EndIndent}

\usepackage{enumitem}
\usepackage{xcolor}
\usepackage{mdframed}
\usepackage{soul}

\usepackage{qcircuit}

\usepackage[hidelinks]{hyperref}

\newtheorem{theorem}{Theorem}[section]
\newtheorem{lemma}{Lemma}[section]

\newtheorem{corollary}{Corollary}[section]
\newtheorem{definition}{Definition}[section]

\def\angstrom{Å}
\newcommand{\prlt}{\xi}

\newcommand{\BL}{\mathcal{B}}

\newcommand{\Ls}{\mathscr{L}}
\newcommand{\conj}[1]{\overline{#1}}
\newcommand{\bw}{W}
\newcommand{\bws}{W_s}
\newcommand{\Hil}{\mathscr{H}}

\begin{document}

\title{Ground and excited-state energies with analytic errors and short time evolution on a quantum computer}

\author{Timothy Stroschein}
\affiliation{ETH Zurich, Department of Chemistry and Applied Biosciences, Vladimir-Prelog-Weg~2, 8093 Zurich, Switzerland}

\author{Davide Castaldo}
\affiliation{ETH Zurich, Department of Chemistry and Applied Biosciences, Vladimir-Prelog-Weg~2, 8093 Zurich, Switzerland}

\author{Markus Reiher}
\email{mreiher@ethz.ch}
\affiliation{ETH Zurich, Department of Chemistry and Applied Biosciences, Vladimir-Prelog-Weg~2, 8093 Zurich, Switzerland}

\date{April 19, 2025}

\begin{abstract}
Accurately solving the Schr\"odinger equation remains a central challenge in computational physics, chemistry, and materials science. Here, we propose an alternative eigenvalue problem based on a system's autocorrelation function, avoiding direct reference to a wave function. In particular, we develop a rigorous approximation framework that enables precise frequency estimation from a finite number of signal samples. Our analysis builds on new results involving prolate spheroidal wave functions and yields error bounds that reveal a sharp accuracy transition governed by the observation time and spectral density of the signal. These results are very general and thus carry far.  As one important example application we consider the quantum computation for molecular systems. By combining our spectral method with a quantum subroutine for signal generation, we define quantum prolate diagonalization (QPD) --- a hybrid classical-quantum algorithm. QPD  simultaneously estimates ground and excited state energies within chemical accuracy at the Heisenberg limit. An analysis of different input states demonstrates the robustness of the method, showing that high precision can be retained even under imperfect state preparation.
\end{abstract}

\maketitle

\section{Introduction}

The efficient solution of eigenvalue problems for large matrices\cite{langhoff1974configuration, van2003iterative, heller2008inflationary} is a foundational task in many areas of natural\cite{bathe1973solution,pederson2022large} and social sciences\cite{brin1998anatomy, ramachandran2005support}. In recent years, significant efforts have been directed toward mapping an abundance of eigenvalue problems onto quantum Hamiltonians\cite{stamatopoulos2024derivative, babbush2023exponential, li2018quantum, pravatto2021quantum}, 
enabling the exploitation of quantum computers for their solution. This surge of interest stems from the initial insight that a programmable quantum device could prevail over the exponential complexity that hampers classical simulations of quantum systems\cite{benioff1980computer, feynman1986quantum}.
Looking more closely at problems in (bio)chemistry\cite{baiardi2023quantum, mcardle2020quantum, goings2022reliably}, catalysis\cite{von2021quantum}, and condensed matter physics\cite{farajollahpour2025quantum} and materials science\cite{liu2022prospects, bauer2020quantum}, realizing these advantages\,\cite{reiher2017elucidating} 
remains challenging due to limitations in current quantum hardware. In particular, despite progress in architecture design\,\cite{vigneau2025quantum, bartolucci2023fusion, aghaee2025scaling, pino2021demonstration}, the necessity of error correction\,\cite{lidar2013quantum} introduces substantial overheads\,\cite{fowler2012surface, lee2024low, gidney2024magic}.

These constraints are especially evident in the implementation of the quantum phase estimation (QPE) algorithm\,\cite{cleve1998quantum}, which is known for achieving the Heisenberg limit in precision scaling\,\cite{giovannetti2006quantum}. Despite its theoretical appeal, even the most recent estimates obtained for the simulation of molecular systems show that QPE suffers from prohibitively high gate counts and qubit requirements\,\cite{dalzell2023quantum}. 
These estimates originate from three key factors: (i) high connectivity requirements due to the multicontrolled evolutions between two different qubit registers make the circuits deep and fragile with respect to noise\,\cite{nelson2024assessment}, (ii) it is necessary to prepare quantum states with sufficiently high overlap with the target state
\cite{erakovic2025high, fomichev2024initial, ollitrault2024enhancing, lee2023evaluating},
and (iii) despite their beneficial formal scaling with system size\,\cite{berry2014exponential, low2019hamiltonian, somma2025shadow}, leading time-evolution methods still demand excessively deep circuits.

These factors have motivated the exploration of alternative approaches that rely on shallower circuits. The first example of a quantum algorithm using a single-ancilla quantum circuit (i.e., the Hadamard test) was given by Kitaev\,\cite{kitaev1997quantum}. Since then, many advances have been made to develop faster and more accurate algorithms, which we refer to as \textit{post}-Kitaev 
approaches\,\cite{ding2023even,ding2023simultaneous,li2023adaptive,stroeks2022spectral,yi2024quantum,gunther2025phase,dutkiewicz2022heisenberg,dutkiewicz2024error}. In particular, the current state-of-the-art single-ancilla phase estimation algorithms are able to satisfy the Heisenberg limit (i.e., $T_{\text{runtime}} = \mathcal{O}(\epsilon^{-1})$) while requiring shorter maximal time evolutions $T_{\text{max}}$ compared to standard QPE. The latter exhibits a scaling of $T_{\text{max}} = \mathcal{O}(\epsilon^{-1})$ and satisfies $T_{\text{max}} \epsilon \approx \pi$, whereas the very recent result of Ref.\,\cite{ding2024quantum} demonstrates a scaling of $T_{\text{max}} = \mathcal{O}(\epsilon^{-1.5})$ with $T_{\text{max}} \epsilon \ll \pi$. 

The methods mentioned above\,\cite{ding2023even, ding2023simultaneous, li2023adaptive, stroeks2022spectral, yi2024quantum, gunther2025phase, dutkiewicz2022heisenberg, dutkiewicz2024error, ding2024quantum} 
are particularly intriguing because they often rely on a formal link between QPE and the signal recovery problem. This connection is also the first key ingredient of the present work and will be discussed in depth in the next section.
A second key ingredient is a thorough exploration of the relationship between signal recovery and subspace methods. A clear equivalence between the two has been established in Refs.\,\cite{neuhauser_bound_1990, mandelshtam_fdm_2001}, where filter diagonalization has been introduced. In the context of quantum algorithms development the connection between QPE and subspace methods has been first made in  Refs.\,\cite{klymko2022real, shen2023real, cortes2022quantum} while Refs.\,\cite{parrish2019quantum,cohn2021quantum} first developed a quantum subspace method inspired by filter diagonalization. 

In this work, we analyze subspace methods through the lens of a post-Kitaev algorithm, characterizing $T_{\text{max}}$, the number of samples and shots per sample needed to achieve a given error estimate $\epsilon$; this perspective differs from previous works in the literature\,\cite{epperly2022theory, shen2023real, kirby2024analysis}. We make use of a new framework to characterize errors in the accuracy and dimensionality of subspace methods\,\cite{stroschein2024prolatespheroidalwavefunctions, stroschein2025approximationframeworksubspacebasedmethods} to derive tight error bounds for an improved version of the filter diagonalization method. 
Our framework builds upon and significantly expands the field of prolate Fourier theory\,\cite{ProI, ProII, ProIII, ProIV, ProV, OnBandwidth, landau_density_1993, SlepianComment}.

With regard to filter diagonalization, we highlight Ref.\,\cite{levitina_filter_2009}, which was the first to employ prolates as filter functions and which made significant progress in detecting the correct number of frequencies within a band.
We build upon their work by introducing new truncation estimates and providing several results that formalize the approach. On the numerical side, the implementation in Ref.\,\cite{levitina_filter_2009} did not account for catastrophic cancellation effects in the generation of prolate functions, a problem, which was addressed and resolved in Ref.\,\cite{Buren2002AccurateCO} and accounted for in our implementation.

All our contributions are enabled by key results in prolate Fourier theory, which we employed to provide a rigorous error analysis of the estimated eigenvalues. We note that these results are of general significance and will therefore have a broader impact, beyond their immediate application to quantum algorithms.
Eventually, we introduce sampled prolate filter diagonalization (PFD) for extracting spectral information from a finite number of signal samples. By generating these samples on a quantum computer via the Hadamard test, we arrive at a hybrid classical-quantum algorithm which we call quantum prolate diagonalization (QPD).

This manuscript is organized as follows: In Section\,\ref{theory}, we develop the theoretical framework underlying this work and summarize the main results. Specifically, the core idea is to shift focus from solving the time-independent Schr\"odinger equation to a related eigenvalue problem involving the autocorrelation function $C(t)= \langle \Psi(t)|\Psi(0) \rangle$.
This eigenvalue problem is the conceptual bridge between post-Kitaev QPE methods and subspace techniques. We solve this spectral problem using prolate filter diagonalization  (Section\,\ref{sec:PFD}), which we review within the general framework of subspace protocols (Section \,\ref{sec:SubspaceProtocols}). Notably, in Section\,\ref{sec:PFD} we unveil a phase transition-like behavior of the accuracy of the protocol crucially depending on the frequency density within the spectral region of interest. 
Subsequently, we discretize prolate filter diagonalization to enable efficient computations which only require a finite number of signal samples (Section\,\ref{sec:sampledPFD}). In Sections\,\ref{sec:sampledPFD}–\ref{sec:prolate_sampling}, we introduce improved truncation bounds for the prolate sampling formula, originally proposed by Walter and Shen\,\cite{walter_sampling_2003}.
Section\,\ref{numerics} presents the numerical results of our QPD implementation applied to the electronic Hamiltonian. We analyze the accuracy scaling behavior of the protocol for both weakly and strongly correlated molecules. We show that our method reaches chemical accuracy (i.e., the error $\epsilon \leq 1\,\mathrm{mHartree}$) and satisfies the Heisenberg limit (Sections\,\ref{sec:mee_estimate}–\ref{sec:heisenberg_limit}). Finally, we investigate the impact of the initial quantum state on performance of QPD to assess the potential quantum advantage offered by our approach (Section\,\ref{input_state_study}).

\section{Theory\label{theory}}

In this section, we introduce our framework for rigorous frequency approximation of signals with dense spectra requiring only a finite number of equidistant samples. Our method facilitates efficient digital computation and provides sharp guarantees on the frequency estimates. Notably, we uncover a phase-transition-like behavior in accuracy: as an \emph{effective spectral density} approaches a critical threshold $T/\pi$
the error increases sharply. As this critical threshold is proportional to the observation time, this transition resembles the Heisenberg uncertainty relation. 

Although the mathematical structure underlying our approach is closely related to the mathematical formulation of quantum mechanics \cite{von_neumann_1927, von_neumann_1927_wahrscheinlichkeit, von_neumann_1929},
we stress that we shift from a traditional wave packet description to a perspective that puts system-specific signals, 
\begin{equation}
C(t) = \int e^{i \lambda t} d \alpha(\lambda),   \label{eq:SigMeasRep}
\end{equation}
to the center of attention. Here, $\alpha$ is a probability measure on $\mathbb{R}$ \cite{bochner1932fourier}. In the standard formalism, such signals arise as quantum autocorrelation functions $C(t)= \langle \Psi | e^{iHt} |\Psi \rangle$.
Furthermore, we assume that the signal is only known within the time interval $[-2T, 2T]$. 

To determine frequencies of a system we do not solve the time-independent Schr\"odinger equation. Instead, we turn to the convolution-based eigenvalue problem
\begin{align}
 - i \partial_\tau (C \ast f) = E (C \ast f).
 \label{eq:SiganlSchrö1}
\end{align}
Here, $f$ is a function in $\mathscr{L}^2(\mathbb{R})$, $E \in \mathbb{R}$, and 
$C \ast f$ denotes the convolution of $C$ with $f$ evaluated in the time variable $\tau$.
This novel spectral equation finds a purely mathematical reason to justify its existence: it facilitates a spectral analysis of signals without assuming an infinite observation time. In this way, it overcomes the greatest limitation of the Fourier transform.
The main result of this work is a mature mathematical framework to approximate solutions of Eq.\ \eqref{eq:SiganlSchrö1}. This eigenvalue  problem is still defined on an infinite-dimensional Hilbert space. It therefore requires approximations for practical computation, which raises fundamental questions on the optimal relationship between observation time, spectral complexity, computational dimensionality and achievable accuracy. 

A seminal series of papers by Slepian, Landau, and Pollak has reached deep insights towards achieving such optimal approximation \cite{ProI, ProII, ProIII, ProIV, ProV, OnBandwidth, landau_density_1993, SlepianComment}. These authors discovered the remarkable properties of prolate spheroidal wave functions, which allowed them to prove that the space of $W$-band-limited and $T$-time concentrated functions is well approximated in essentially $2WT/\pi$ degrees of freedom. 
We refer to the field that their series of papers has initiated as prolate Fourier theory (PFT).

In this work we apply and extend the rich set of mathematical properties of prolate spheroidal wave functions to develop a method for frequency estimation that can access fundamental approximation bounds. We use the commutation relation of an integral operator and differential operator, to derive new concentration identities of prolates. The new identities enable sharp accuracy guarantees for prolate filter diagonalization (PFD). They further allow us to give a sampling theorem that puts mathematical truth to the long standing intuition that approximately $2WT/\pi$ samples suffice to reconstruct a signal of bandwidth $W$ and high concentration in the time interval $[-T,T]$ with high accuracy. The prolate sampling formula is applied to achieve a discretization of PFD, that only assumes a finite number of samples of the signal and guarantees high accuracy. 
To clarify the structure of our theory, Figure~\ref{fig:dependency_graph} shows the logical dependencies between key results and developments within our framework.

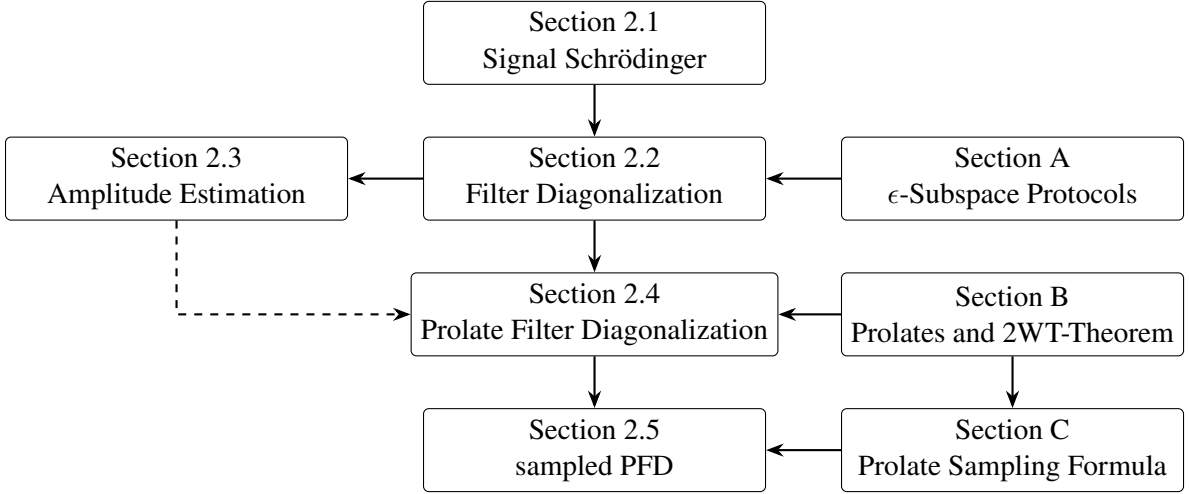
\begin{figure}[H]

\centering
\begin{tikzpicture}[
  node distance=1.6cm and 4.2cm,
  every node/.style={draw, rounded corners=2pt, minimum width=4.5cm, minimum height=1.1cm, align=center},
  arrow/.style={-{Stealth}, thick},
  dashedarrow/.style={-{Stealth}, thick, dashed}
]

% Column positions
\def\xleft{-5.5}
\def\xcenter{0}
\def\xright{5.5}

% Y coordinates
\def\yA{3.2}   % Signal Schrödinger
\def\yB{1.4}   % Filter Diagonalization
\def\yC{-0.4}  % PFD
\def\yD{-2.2}  % Sampled PFD

% Nodes
% Left column
\node (Amp) at (\xleft,\yB) {Section \ref{sec:Amp}\\ Amplitude estimation};

% Center column (main chain)
\node (SigSchro) at (\xcenter,\yA) {Section \ref{sec:sig_schrö}\\ Signal Schrödinger};
\node (FD)       at (\xcenter,\yB) {Section \ref{sec:FilterDiag}\\ Filter diagonalization};
\node (PFD)      at (\xcenter,\yC) {Section \ref{sec:PFD}\\ Prolate filter diagonalization};
\node (sPFD)     at (\xcenter,\yD) {Section \ref{sec:sampledPFD}\\ Sampled PFD};

% Right column
\node (Subspace)         at (\xright,\yB) {Appendix \ref{sec:SubspaceProtocols}\\ Subspace protocols};
\node (Prolates)         at (\xright,\yC) {Appendix \ref{sec:prolates}\\ Prolates and 2WT-Theorem};
\node (SampledProlates)  at (\xright,\yD) {Appendix \ref{sec:prolate_sampling}\\ Prolate sampling formula};

% Arrows (center chain)
\draw[arrow] (SigSchro) -- (FD);
\draw[arrow] (FD) -- (PFD);
\draw[arrow] (PFD) -- (sPFD);

% Arrows (right to center)
\draw[arrow] (Subspace) -- (FD);
\draw[arrow] (Prolates) -- (PFD);
\draw[arrow] (Prolates) -- (SampledProlates);
\draw[arrow] (SampledProlates) -- (sPFD);

% New arrows involving Amplitude Estimation
\draw[arrow] (FD.west) -- (Amp.east);
\draw[dashedarrow] (Amp.south) |- (PFD.west);

\end{tikzpicture}

\caption{Logical dependency graph of the theory underlying our approach: for the main part of the paper we keep the essential building blocks of our approach, while we cover additional key parts of the theoretical framework in the appendix. The relation between these parts is highlighted in this figure to aid the understanding of our work. Section \ref{sec:Amp} has indirect implications for Section \ref{sec:PFD}.}
\label{fig:dependency_graph}
\end{figure}

% \subsection{From the Schrödinger Equation to Signal Processing}
\subsection{From signals to quantum dynamics: an inverse approach}
\label{sec:sig_schrö}
Within quantum mechanics, a system is described by a self-adjoint operator $H$ that acts on a Hilbert space $\mathscr{H}$. 
The states of the system are described by time-dependent wave functions $|\Psi(\cdot) \rangle : \mathbb{R} \rightarrow \mathscr{H} $, which evolve according to the time-dependent Schr\"odinger equation:
\begin{align}
    i \partial_t |\Psi(t) \rangle = H |\Psi(t) \rangle.
    \label{eq:Schrö}
\end{align}
For a time-independent Hamiltonian $H$, stationary states are solutions of the form $|\varphi_n(t) \rangle = e^{-iE_n t} |\varphi_n \rangle$, where $|\varphi_n \rangle$ is an eigenvector of $H$,
\begin{align}
    H |\varphi_n \rangle= E_n |\varphi_n \rangle.
    \label{eq:timeindep}
\end{align}
They satisfy in particular
\begin{align}
    i \partial_t |\varphi_n (t) \rangle= E_n |\varphi_n(t) \rangle
\label{eq:timedep}
\end{align}
and their probability distribution remains constant in time. The time evolution of a general wave function is described by its decomposition into stationary states
$|\Psi(t) \rangle = \sum a_n e^{-iE_n t} |\varphi_n\rangle$. 
While fundamental for encoding physical laws, this formalism relies on a high-dimensional 
(precisely, on an infinite-dimensional) Hilbert space and the numerical representation of $H$ and state vectors suffers from the curse of dimensionality. By contrast, experimental observations consist of low-dimensional quantities such as expectation values, low-lying eigenvalues, transition rates, and magnetization. This discrepancy in dimensionality between experimental observation and mathematical description causes a large overhead in computational applications.

To address this imbalance, we adopt an approach focusing directly on signals rather than wave packets.
Specifically, we consider a continuous positive definite function $C(t)$ with $C(0) =1 $ and aim to determine its frequencies while only assuming a finite observation time. In this framework, we do not postulate the wave packet formalism or the Schrödinger equation; instead, these emerge as mathematical auxiliaries to analyze the signal.

We start by justifying Eq.\,\eqref{eq:SiganlSchrö1}
to determine the frequencies of $C$. By Bochner's theorem \cite{bochner1932fourier}, any continuous positive definite function $C$ is characterized by a unique measure $\alpha: \mathcal{B}(\mathbb{R}) \to \mathbb{R}$ such that Eq.\,\eqref{eq:SigMeasRep} holds. We define the spectrum of $C$ as $\sigma(C) := \operatorname{supp}(\alpha)$.
Let $F$ denote the Fourier transform of a function $f\in \mathscr{L}^2(\mathbb{R})$. Inserting Eq.\,\eqref{eq:SigMeasRep} 
into Eq.\,\eqref{eq:SiganlSchrö1} 
and interchanging the integrals yields
\begin{align}
     \int  \lambda e^{i \lambda \tau } F(\lambda) d\alpha(\lambda) \stackrel{!}{=} E  \int  e^{i \lambda \tau } F(\lambda) d\alpha(\lambda).
\end{align}
This equality holds for all $\tau$ if and only if the complex measure $(\lambda - E)F(\lambda)d\alpha(\lambda)$ is zero. For a non-trivial solution, this implies that $E \in \sigma(C)$ and the $\alpha$-essential support of $F$ is the singleton $\{E\}$, i.e., $\operatorname{supp}_{\alpha}(F)=\{E\}$.
For a purely discrete signal $C(t) = \sum_k |a_k|^2 e^{i E_k t}$, this simplifies to the condition:
\begin{align}
    F(E_n) \neq 0 \qquad \text{ and } \qquad F(E_k) = 0  \text{ for all }  E_k \neq E_n.
    \label{eq:zero_conditions}
\end{align}
Hence, the eigenvalue problem in Eq.\,\eqref{eq:SiganlSchrö1} 
corresponds to finding a function with a Fourier transform whose roots coincide with all frequencies except one.

To perform a frequency analysis without assuming an infinite observation time, we restrict $f$ to the subspace $\mathscr{D}_T \subset \mathscr{L}^2(\mathbb{R})$ of functions supported on the interval $[-T,T]$. Consequently, the convolution in Eq.\,\eqref{eq:SiganlSchrö1} 
only requires access to $C(t)$ for a finite duration, allowing us to overcome a major limitation of the standard Fourier transformation.
A non-trivial solution to Eq.\,\eqref{eq:zero_conditions} 
exists in $\mathscr{D}_T$ if and only if the spectrum of $C$ is discrete and has no accumulation point. This follows from the  identity theorem and the fact that the Fourier transform of a time-limited function is entire.

While the restriction to $\mathscr{D}_T$ addresses the finite observation time, directly identifying a function $f$ with the requisite zeros in the Fourier transform remains a challenging continuous search problem. In order to formulate an efficient numerical eigensolver, it is advantageous to map this signal processing task onto the linear algebraic structure of quantum mechanics.

To this end, we note that any positive definite function $C$ with $C(0)=1$ can be interpreted as a correlation function arising from unitary time evolution on a Hilbert space. That is, we can construct a Hilbert space $\mathscr{H}$, a self-adjoint operator $H$ acting on $\mathscr{H}$, and a vector $\Psi \in \mathscr{H}$, such that the spectra of the signal and operator $H$ coincide, $\sigma(C) = \sigma(H)$, and
\begin{align}
C(t) = \langle \Psi, e^{i H t} \Psi \rangle = \int_{\mathbb{R}} e^{i \lambda t} \, d\langle \Psi, P^{(H)}(\lambda) \Psi \rangle. \label{eq:QMinterpret}
\end{align}
Here $P^{(H)}$ denotes the unique projection-valued spectral measure with $H = \int \lambda \, dP^{(H)}(\lambda)$. For a signal with a purely discrete spectrum $C(t) =\sum_k |a_k|^2 e^{i E_k t}$, each frequency $E_k$ is associated with an eigenvector $|\varphi_k \rangle$ of $H$ and a valid construction has $|\Psi\rangle = \sum_k a_k |\varphi_k\rangle$. In particular, $H, |\Psi\rangle$, and $\{ \varphi_k \}$ satisfy the Schr\"odinger equation \eqref{eq:Schrö}--\eqref{eq:timedep} by construction. Hence, we do not postulate the wave packet formalism, but instead construct it as a mathematical auxiliary.
Crucially, this construction remains implicit, effectively hiding the redundant degrees of freedom arising from unitary invariance in explicit basis expansions of $H$. This circumvents the curse of dimensionality of operator representations while preserving access to the relevant spectral data. In Section\,\ref{sec:FilterDiag}
, this perspective provides the rigorous foundation to formulate  filter diagonalization as a subspace protocol (appendix\,\ref{sec:SubspaceProtocols}).

\subsection{Filter diagonalization}
\label{sec:FilterDiag}

Filter diagonalization is a numerical procedure to approximate the frequencies of a signal,
\begin{equation}
    C(t) = \sum_n |a_n|^2 e^{iE_n t},
\end{equation}
within a target band $B_{f,\omega^*}=[-W_f + \omega^*,W_f + \omega^*]$. Without loss of generality, we assume the band is centered at zero (otherwise, we shift the signal $C(t) \mapsto C(t) e^{-i \omega^* t}$).

This method fits naturally into the framework of \emph{subspace protocols}~\cite{stroschein2025approximationframeworksubspacebasedmethods}. A subspace protocol approximates the eigenvalues of an operator $H$ within a specific spectral subspace $\mathcal{E}$ by constructing a set of trial vectors $\{v_l\}$ that span $\mathcal{E}$ (appendix~\ref{sec:SubspaceProtocols}
reviews the structure and algorithmic operation of subspace protocols and restates the master theorem on their accuracy). We construct these trial vectors within the implicit quantum mechanical structure established in Section \ref{sec:sig_schrö}, where the signal is assumed to be generated by a state $|\Psi\rangle = \sum_k a_k |\varphi_k\rangle$ evolving under a Hamiltonian $H$. However, the final protocol requires only the signal $C(t)$ and no explicit access to $H$ or $|\Psi\rangle$ (for a comprehensive derivation see Ref.~\cite[Section 5.2]{stroschein2024prolatespheroidalwavefunctions}).

To target the spectral subspace $\mathcal{E}$ corresponding to the frequency band $B_f$, we apply a set of filter functions $\{f_l(t)\}_{l=0}^{m-1}$ to the trajectory of the state $|\Psi(t)\rangle$. Explicitly, the $l$-th trial vector is generated by the integral transform:
\begin{align}
v_l := \int_{-\infty}^{\infty} f_l(\tau) |\Psi(\tau)\rangle d\tau. \label{eq:trial_vec_integral}
\end{align}
Substituting the eigen-decomposition $|\Psi(\tau)\rangle = \sum_k a_k e^{-i E_k \tau} |\varphi_k\rangle$ into Eq.\,\eqref{eq:trial_vec_integral} yields:
\begin{align}
v_l = \sum_k a_k \left( \int_{-\infty}^{\infty} f_l(\tau) e^{-i E_k \tau} d\tau \right) |\varphi_k\rangle = \sum_k a_k F_l(E_k) |\varphi_k\rangle, \label{eq:trial_vec_FD}
\end{align}
where $F_l(E)$ is the Fourier transform of the filter function $f_l(t)$. Finally, filter diagonalization can be formulated as a subspace protocol (Definition~\ref{def:SubspaceProtocol}; see appendix) with matrix elements
\begin{align}
    (V^\dagger H V)_{sl} &=  -i \int f_s(\tau)\, \partial_\tau \big(C \ast f_l\big)(\tau) \, d\tau,  \label{eq:FD_A_sl} \\ 
    (V^\dagger V)_{sl} &= \int f_s(\tau)\, \big(C \ast f_l\big)(\tau) \, d\tau,  \label{eq:FD_B_sl} \\
    \delta A &= \delta B = 0.
\end{align}
Given a signal $C$ and a target band $B_f$, Algorithm~\ref{alg:eps_dim_red} 
defines a numerical routine $\mathsf{P}_{(C, B_f)}$ to extract the frequencies within $B_f$ (cf. appendix \ref{sec:SubspaceProtocols}). The routine $\mathsf{P}_{(C, B_f)}$ takes the number $M$ of trial vectors to use and the assumed number $m$ of frequencies inside the band $B_{f}$ as input.
Step~2 of Algorithm~\ref{alg:eps_dim_red} can be used to estimate $m$ through the spectrum of the weight matrix $B_M$ and a noise threshold $\epsilon_{\text{th}}$. 
Ideally, the spectrum of $B_M$ exhibits a sharp drop, such that the number of significant frequencies inside the band $B_f$ is detected with high confidence.
Step~3 of Algorithm~\ref{alg:eps_dim_red} performs a refinement of the generalized eigenvalue problem $[A_M,B_M]$, to obtained a well-conditioned GEP $[A_M^m, B_M^m]$ of dimension $m$. The refinement allows us to use an initial trial dimension $M$ that exceeds $m$. A larger $M$ increases the expressiveness of the trial vector space and improves the conditioning of the refined GEP. This can also enable a more reliable detection of $m$.
However, there is a trade-off as contributions from noise typically also increase with $M$. 
So, the choice of initial trial dimension $M$ must be balanced between improved signal representation and increased distortions from noise (see Ref.\cite[Section~3.1]{stroschein2025approximationframeworksubspacebasedmethods}).

After the correct number of frequencies has been detected, Theorem~\ref{thm:subspace_master} (see appendix)
describes the accuracy of the frequency estimates computed by $\mathsf{P}_{(C, B_f)}(M, m)$. The approximation error is quantified through an error measure $\varepsilon_M: \mathcal{B}(\mathbb{R}) \to \mathbb{R}$, which, in the case of filter diagonalization, is given by
\begin{align}
    \varepsilon_M(I) &= \sum_{i=0}^{M-1} \int_{I} |F_i(E)|^2 \, d\alpha(E) \label{eq:filter_meas} \\
    &= \sum_{i=0}^{M-1} \sum_{E_k \in I} |F_i(E_k)|^2 |a_k|^2. \notag
\end{align}

\subsection{Amplitude estimation}
\label{sec:Amp}

Once the frequencies have been approximated, filter diagonalization can also recover the corresponding amplitudes. We introduce a novel approach that facilitates accuracy guarantees for its amplitude estimates. In addition, we give a bound on the minimal amplitude required for a frequency to distinguish itself from a noise threshold.
Early implementations of filter diagonalization extracted amplitudes from the eigenvectors of the GEP \cite{wall_extraction_1995, mandelshtam_harmonic_1997, mandelshtam_fdm_2001}. However, this approach remained heuristic and lacked a rigorous error or stability analysis.

Our approach is based on the observation that the filter diagonalization matrices can be decomposed into matrix multiplications of alternant matrices and diagonal matrices containing the amplitudes \cite[Section~5.2 \& 5.3]{stroschein2024prolatespheroidalwavefunctions}: We let $\vec{E} \in \mathbb{R}^{m}$ denote the vector of frequencies inside the target band $B_f$ and define the \emph{alternant matrix} $\mathfrak{F}_M(\vec{E}) \in \mathbb{C}^{m \times M}$ by
\begin{align}
    \mathfrak{F}_M(\vec{E}) := \big(F_l(E_k)\big)_{k l}, \label{eq:alternant_full}
\end{align}
where $F_l$ is the Fourier transform of the $l$-th filter function. 
Consider the refined GEP $[A_M^m, B_M^m]$ returned by filter diagonalization through Algorithm~\ref{alg:eps_dim_red}.
The weight matrix is found to have the form
\begin{align}
    B_M^m = \mathfrak{F}_M^m(\vec{E})^\dagger \Lambda(\vec{a}) \, \mathfrak{F}_M^m(\vec{E}) + (\mathcal{N}^{(B)})_M^m, \label{eq:refinedweight}
\end{align}
where $\Lambda(\vec{a}) = \mathrm{diag}(|a_k|^2)$ is the amplitude matrix, $\mathfrak{F}_M^m(\vec{E}) := \mathfrak{F}_M(\vec{E}) U_m \in \mathbb{C}^{m \times m}$ is the alternant matrix formed from the refined set of filter functions obtained after step~$3$ of Algorithm~\ref{alg:eps_dim_red}, and
$(\mathcal{N}^{(B)})_M^m := U_m^\dagger \big(N^\dagger_M N_M + \delta B_M\big) U_m$ is the projected noise weight matrix.
This representation of the weight matrix follows from rewriting sums over the frequency contributions as a matrix multiplication.

Eq.\ \eqref{eq:refinedweight} provides a practical method to approximate the amplitudes that also enables accuracy guarantees. Let 
$\tilde{E}$ denote the vector of eigenvalues obtained by filter diagonalization. Since $\tilde{E}$ approximates the true frequency vector $\vec{E}$, we estimate the amplitude matrix as
\begin{align}
    \Lambda(\vec{a}) \approx \mathfrak{F}_M^m(\tilde{E})^{-\dagger} \, B_M^m \, \mathfrak{F}_M^m(\tilde{E})^{-1}.  \label{eq:Amp_estimation}
\end{align}
Proposition~7 of Ref.\cite[Section~5.3.2]{stroschein2024prolatespheroidalwavefunctions} bounds the deviation of the diagonal elements of\\
$\mathfrak{F}_M^m(\tilde E )^{-\dagger} B_M^m \mathfrak{F}_M^m(\tilde{E})^{-1}$ to the true amplitudes. Moreover, one can assess the convergence of the method, by examining whether this matrix is indeed close to diagonal.

The alternant matrix also encodes meaningful information on the stability of the method. 
In particular, we establish a novel amplitude threshold that guarantees that all frequencies can be distinguished from a given noise level:
\begin{equation}
    |a|^2_{\min} > \lambda_{\min}\!\bigl(\mathfrak{F}_M^m(\vec{E})^\dagger \mathfrak{F}_M^m(\vec{E})\bigr)^{-1} \epsilon_{\mathrm{th}} . \label{eq:minimalamplitude}
\end{equation}
If all relevant amplitudes satisfy Eq.\ \eqref{eq:minimalamplitude}, Algorithm~\ref{alg:eps_dim_red} is indeed guaranteed to detect the correct number of frequencies within the band:
\begin{align}
    \lambda_m(B_M^m)
    &\;\ge\; \lambda_m\!\bigl(\mathfrak{F}_M^m(\vec{E})^\dagger \mathfrak{F}_M^m(\vec{E})\bigr)\,|a|^2_{\min} \;+\; \lambda_m\bigl((\mathcal{N}^{(B)})_M^m\bigr)  \label{eq:weight_lowerbound}  \\
    &\;>\; \epsilon_{\mathrm{th}}.
\end{align}
In the last inequality, we assume that $(\mathcal{N}^{(B)})_M^m$ is positive definite, as is the case for purely subspace-based noise $\mathcal{N}^{(B)}_M = N^\dagger_M N_M$.

\subsection{Prolate filter diagonalization}
\label{sec:PFD}

In this section we review prolate filter diagonalization (PFD) and its error bounds \cite[Chapter 5]{stroschein2024prolatespheroidalwavefunctions}. We highlight a
phase-transition-like behavior in the accuracy of PFD that resembles the Heisenberg uncertainty relation: the accuracy drastically decreases as the spectral density of the signal approaches a critical threshold given by the observation time $\sim T/\pi$.

PFD uses time-limited prolates $\mathcal{D}_T \prlt_n^{(W_f)}$ as filtering functions, due to their optimal concentration in the frequency band $B_f$ and favorable orthogonality relations. They are eigenfunctions of the time- and band-limiting integral operator $\mathcal{D}_T \mathcal{B}_{f}$ 
(see appendix \ref{sec:prolates}). 
In particular, there are
approximately $2W_fT/\pi$ such time-limited prolates with high concentration in the frequency band $B_{f}$. We deferred a review of the origin of prolates and their optimal approximation properties to appendix \ref{sec:prolates} and note that in
\ref{sec:newIneq}, we review new inequalities on prolates that enable rigorous accuracy guarantees for PFD.

We formalize PFD as a subspace protocol that can access the theory 
of appendix \ref{sec:SubspaceProtocols}:

\begin{definition}[Prolate Filter Diagonalization]
\label{def:PFD}
    Given a signal $C(t)$ for $t \in [-2T, 2T]$ and a frequency band $B_{f}=[-W_f, W_f]$ of interest. We define \emph{Prolate Filter Diagonalization} (PFD) as the protocol $\mathsf{P}$ that assumes infinite dimensional matrices $V^\dagger H V$ and  $V^\dagger V$ with matrix elements:
    \begin{align}
    ( V^\dagger H V)_{sl} & = - i \int_{-T}^{T} \int_{-T}^{T} \prlt_s^{(W_f)}(\tau) \partial_\tau C(\tau - t) \prlt_l^{(W_f)}(\tau) d t d\tau   \label{eq:VHV_sl}\\
    ( V^\dagger V)_{sl} & = \int_{-T}^{T} \int_{-T}^{T} \prlt_s^{(W_f)}(\tau) C(\tau - t) \prlt_l^{(W_f)}(\tau) d t d\tau \label{eq:VV_sl}
    \end{align}
      Here, $\{\prlt_n^{(W_f)} \}$ denotes the sequence of prolates of bandwidth $W_f$ and concentration in the time interval $[-T,T]$.
      At a given trial dimension $M$, PFD returns the $M$-th leading principal sub-matrices of $V^\dagger H V$ and $V^\dagger V$.
\end{definition}
PFD constitutes a subspace protocol as in Definition~\ref{def:SubspaceProtocol}, with  $\delta A = \delta B = 0$ and noise matrix $\mathcal{N}^{(B)} = N^\dagger N $,
\begin{align}
    (N^\dagger N)_{sl} = \int_{-T}^{T} \int_{-T}^{T} \int_{\sigma(C) \backslash B_{f}} \prlt_{s}^{(W_f)}(\tau) e^{i E (\tau - t)} \prlt_l^{(W_f)}(t) d\alpha(E) dt d\tau,
\end{align}
arising from frequency components outside the target band.
Moreover, PFD can be formulated as an $\epsilon$-subspace protocol $\mathsf{P}_{\{\epsilon_M\}}$ (see Definition~\ref{def:eps_SubspaceProtocol}) 
with bounds on the noise weight
\begin{align}
    \| \mathcal{N}^{(B)} \| & \leq \epsilon^{\text{Prlt}}_M \int_{ \sigma(C) \backslash B_{W_f}} d\alpha(\lambda) =  \epsilon^{\text{Prlt}}_M \sum_{E_n \notin B_f} |a_n| e^{iE_n t}\\
    & \leq  \epsilon^{\text{Prlt}}_M C(0),  \label{eq:noise_weigt_bound}
\end{align}
where $\epsilon^{\text{Prlt}}_M$ is defined in Eq.\ \eqref{eq:prlt_err_th} below.

Theorem \ref{thm:subspace_master} 
(cf.\ appendix) 
describes the accuracy of a subspace protocol through an error measure $\varepsilon_M$. For filter diagonalization, this error measure is determined by the Fourier transforms of the filtering functions, as seen in Eq.\ \eqref{eq:filter_meas}. The Fourier transform of a time-truncated prolate is
\begin{align}
    \mathcal{F} [ \mathcal{D}_T \prlt_n](\omega) = \conj{\mu}_n \tilde \prlt_n(\omega),
\end{align}
where $\tilde \prlt_n$ denotes \emph{dual prolates}, that is prolates with roles of bandwidth $W$ and time interval $T$ exchanged. The prefactors $\conj{\mu}_n \in \mathbb{C}$ are complex conjugates of the eigenvalues of the integral operator given in Eq. \eqref{eq:FiniteFT_op} in appendix \ref{sec:prolates}.

The bound in Eq.\ \eqref{eq:thm_prop_bound_extra} of Theorem~\ref{thm:prlt_bound} controls the maximal value that a prolate function can attain outside its concentration region.
We derived Theorem~\ref{thm:prlt_bound} in Ref.\ \cite[Chapter 2]{stroschein2024prolatespheroidalwavefunctions} to yield sharp accuracy bounds for PFD. The proof of Theorem~\ref{thm:prlt_bound} relies on a commutation relation between an integral operator and a differential operator. This commutation relation has played an important role in the development of PFT and continues to lead to remarkable connections \cite{SlepianComment, grunbaum2022matrixbispectralitynoncommutativealgebras, UVSpectrum}.
In appendix \ref{sec:dual_signature}, we highlight a spectral duality between the integral and differential operator at the critical dimension $n = \lfloor 2W T /\pi \rfloor$, which appears to have remained previously unnoticed.

By applying Eq.\ \eqref{eq:thm_prop_bound_extra} on the prolate duals $\tilde \prlt_n(\omega)$, we obtain an estimate on the error measure of PFD: 
\begin{corollary}
\label{cor:meas_estimate}
    Let $\varepsilon_M$ be the error measure that is obtained from the first $M$ time-limited prolates 
    $\mathcal{D}_T \prlt_n^{(W_f)}$. We have
    \begin{align}
        \varepsilon_M(I) \leq 2 \pi \sum_{l=0}^{M-1} \gamma_l (1-\gamma_l) C_{\text{extra},l}(T,W_f)   \ \alpha (I)
    \end{align}
    where $C_{\text{extra},l}(T,W_f)$ is defined as in Theorem \ref{thm:prlt_bound} with arguments $W_f$ and $T$ exchanged. 
\end{corollary}
We abbreviate the right-hand side as
\begin{align}
    \epsilon^{\text{Prlt}}_M & := 2 \pi \sum_{l=0}^{M-1} \gamma_l (1-\gamma_l) C_{\text{extra},l}(T,W_f)  \label{eq:prlt_err_th}
\end{align}
since $\epsilon^{\text{Prlt}}_M$ characterizes the accuracy of PFD. 
The magnitude of $\epsilon^{\text{Prlt}}_M$ is strongly governed by the factor $1-\gamma_{M-1}(c_f)$, where  $c_f:= W_f T$.
We refer to $\tilde{c}_f := 2W_f T / \pi $ as the \emph{essential dimensionality}, since the 2WT-Theorem (Theorem~\ref{thm:2WT} in the appendix) states that the space of $ W_f $-bandlimited and $ T $-time-concentrated functions is well approximated by essentially \( 2W_f T / \pi \) degrees of freedom.
We can estimate 
\begin{align}
    \epsilon^{\text{Prlt}}_M & <  2 \pi M  T c_f ( 1 + \mathcal{O}(c^{-2}_f)) (1-\gamma_{M-1}(c_f))   \label{eq:eps_asymp_estimate_0} \\
    & <M T \frac{\pi^{\frac{3}{2}} 2^{3 M} c_f^{M+\frac{1}{2}} e^{-2 c_f}}{(M-1)!}\left[1+\mathcal{O}\left(c^{-1}_f\right)\right]. \label{eq:eps_asymp_estimate_1}
\end{align}
The first inequality applies bounds on $C_{\text{extra},l}(T,W)$ from appendix \ref{sec:newIneq}, and the second uses the asymptotic expansion
\begin{align}
    1-\gamma_n(c) =  \frac{4 \sqrt{\pi} 2^{3 n} c^{n+\frac{1}{2}} e^{-2 c}}{n!} (1 + \mathcal{O}(c^{-1}))
\end{align}
according to Refs.\cite{FUCHS1964317, SlepianAsymp}. In particular, $\epsilon^{\text{Prlt}}_M$ 
remains very small as long as the number of used prolate filters $M$ remains sufficiently lower than the essential dimension $2W_f T/\pi $. At a fixed target accuracy and trial dimension, $\epsilon^{\text{Prlt}}_M$  and Eq.\ \eqref{eq:eps_asymp_estimate_1} can also be used to estimate the required observation duration $T $ or phase space volume $ c_f = W_f T $.\\

Corollary \ref{cor:meas_estimate} and Theorem \ref{thm:subspace_master} 
together yield the following sharp accuracy guarantee for PFD: 

\begin{theorem}
    \label{thm:PFD_accuracy}
    Let $[A_M^m, B_M^m]$ be the GEP that is obtained from PFD through Algorithm \ref{alg:eps_dim_red} with a trial dimension $M$ and a detected dimension $m$.
    Assume that the detected dimension $m$ coincides with number of frequencies in $B_{f}$, that the GEP is well-conditioned 
    in the sense 
    \begin{align}
        \lambda_m(B_M^m) > \| \mathcal{N}^{(B)}_M \|,
    \end{align}
    and that the generalized eigenvalues lie within $[-W_f, W_f]$.
    Let $\tilde E_1 > \tilde E_2 > \cdots > \tilde E_m$ denote these eigenvalues, and $E_1 > E_2 > \cdots > E_m$ the frequencies of $C$ within $[-W_f, W_f]$.
    Then, we have error bounds on the approximated frequencies: 
    \begin{align}
  \epsilon^{\text{Prlt}}_M \frac{ \int_{\lambda < -W_f} (\lambda - \tilde E_i) d \alpha(\lambda) }{\lambda_m(B_M^m) - \| \mathcal{N}^{(B)}_M \| } \leq \tilde E_i -E_i \leq \epsilon^{\text{Prlt}}_M \frac{\int_{\lambda > W_f} (\lambda - \tilde E_i) d \alpha(\lambda)  }{\lambda_m(B_M^m) -  \| \mathcal{N}^{(B)}_M \|}.
\end{align}
\end{theorem}

Theorem~\ref{thm:PFD_accuracy} provides sharp error bounds for the estimated frequencies. These bounds are governed by the error parameter
$\epsilon^{\text{Prlt}}_M$, which essentially behaves like $1-\gamma_M(c_f)$. This leads to a \emph{phase transition} like behavior in the accuracy of the method:
\emph{Prolate Filter Diagonalization achieves high accuracy within a spectral region of interest, as long as the \emph{effective spectral density} $\delta_{\text{eff}}$ of the signal has 
\begin{align}
    \delta _{\text{eff}}   \lesssim  \frac{T}{\pi}. \label{eq:accuracy_condition_eff}
\end{align}
As $\delta_{\text{eff}}$ approaches $T / \pi$ the accuracy of PFD sharply decreases.}

This statement follows from the number of trial vectors necessary to represent a signal within $B_f$ and $\epsilon^{\text{Prlt}}_M \sim 1 - \gamma_M(c_f)$.
The effective spectral density is slightly more locally defined than the spectral density, as we explain below.

The first $\sim 2c/\pi $ eigenvalues $\gamma_n(c)$ are very close to one, 
followed by a transition region, that is logarithmically sharp in $c$. Beyond the transition region all remaining eigenvalues rapidly decay toward zero (Theorem \ref{thm:prlt_spectrum}).
The characterizing accuracy parameter $\epsilon^{\text{Prlt}}_M$ of PFD inherits this behavior; as $M$ approaches $2W_f T/\pi$ there is a sharp increase in $\epsilon^{\text{Prlt}}_M$.

The number of trial vectors necessary to use depends of the effective spectral density of the signal. A necessary condition to stay out of the transition region in terms of the spectral density is 
\begin{align}
    \frac{\# \{E_k \in B_{f, \omega^*} \}}{2 W_f}  \lesssim  \frac{T}{\pi}. \label{eq:accuracy_condition}
\end{align}
The condition is not sufficient, as the error bounds of Theorem \ref{thm:PFD_accuracy} also depend on the conditioning of $B_M^m$. In Section \ref{sec:Amp} we derived that the conditioning of the weight matrix does not only depend on 
the amplitudes, but also on the conditioning of the filter matrix evaluated at the in-band frequencies:
\begin{align}
        \lambda_m(B_M^m) &\;\ge\; \lambda_m\!\bigl(\mathfrak{F}_M^m(\vec{E})^\dagger \mathfrak{F}_M^m(\vec{E})\bigr)\,|a|^2_{\min}.
\end{align}
Hence we require, for $M$ large enough but lower than $2W_fT/\pi$, that a good conditioning of $B_M^m$ can be ensured. Despite the necessary condition in Eq.\ \eqref{eq:accuracy_condition} being satisfied, the contrary can occur if two frequencies are very close with respect to the analyzed frequency band $B_f$. This can cause near-linear dependencies in the rows of $ \mathfrak{F}_M^m(\vec{E})$. In this case, it is appropriate to "zoom in" on the approximate degeneracy through a more informed choice of $B_f$. 
This also suggest an intuitive definition of the \emph{effective spectral density}: it is defined through $\# \{E_k \in B_{f,\omega^*} \}/ {2 W_f}$ but with a choice of $B_{f,\omega^*}$ such that the spectral density within the band remains relatively homogeneous.

The accuracy transition we just uncovered through PFT is reminiscent of a generalization of the Heisenberg uncertainty principle.
As long as the observation time $T$ satisfies Eq. \ \eqref{eq:accuracy_condition_eff}, PFD achieves high accuracy despite dense spectra with many frequencies outside the target band $B_f$.

\subsection{Sampled prolate filter diagonalization}
\label{sec:sampledPFD}

In this section we discretize PFD through the \emph{prolate sampling formula} to obtain a method particularly suited for digital computation.
This discretization requires an additional assumption on the bandwidth of the signal,
but in return, only relies on a finite number of equidistant samples. \\
We emphasize that the sampling-induced error exhibits a phase-transition behavior, which is simultaneous with the phase transition of the subspace-based error derived in Section~\ref{sec:PFD}.

The prolate sampling formula was introduced in Ref.\ \cite{walter_sampling_2003}, and it improves upon the Whittaker-Shannon sampling formula, because it truncates better. However, the truncation estimates in Ref.\ \cite{walter_sampling_2003} were still relatively loose. Using the inequalities of Theorem~\ref{thm:prlt_bound}, we can derive sharper estimates. 
In appendix \ref{sec:prolate_sampling}, we formulate an approximation theorem that resembles the engineering folk theorem \cite{NoiseCommunication, SlepianComment}:
\emph{Approximately $2WT/\pi$ equidistant samples suffice to reconstruct a signal of bandwidth $W$ and essential duration $2T$ with high accuracy. }

This theorem was a long standing intuition that preceded  
the $2WT$ theorem (see Section \ref{sec:prolate_sampling} and references therein). 
Let $\mathscr{B}_s \subset \mathscr{L}^2(\mathbb{R})$ denote the subspace of functions with a Fourier transform that has support in $[-W_s,W_s]$.
Given a band-limited function $f\in\mathscr{B}_{s}$ that is highly concentrated in the interval $[-T,T]$,
we approximate $f$ on $[-T,T]$ through:
\begin{align}
    f(t) \approx \frac{\pi}{W_s} \sum_{n=0}^{\lfloor 2 W_s T/\pi \rfloor} \sum_{k = - \lfloor W_s T/\pi\rfloor }^{\lfloor W_s T/\pi\rfloor} f\left(\frac{k \pi}{\bw}\right) \prlt_n^{(W_s)}\left(\frac{k \pi}{\bw}\right) \prlt_n^{(W_s)}(t),\label{eq:trunc_prlt_sampling_2WT}
\end{align}
 where $\{\prlt_n^{(W_s)} \}$ denotes the sequence of prolates of bandwidth $W_s$ and concentration in $[-T,T]$. The accuracy of the approximation is described by Theorem~\ref{thm:SamplingEngFolk} in the appendix.

To apply this truncated sampling formula to PFD, we require the additional assumption that the signal $C$ is of finite bandwidth $W_C$.
Then, the integrands appearing in the matrix elements of PFD,
\begin{align}
 \label{eq:functiontosample}
f_l^{(\tau)}(t) := C(\tau - t) \prlt_l^{(W_s)}(t),
\end{align}
have bandwidth at most $W_C + W_f$. And note that $f_l^{(\tau)}(t)$ is concentrated in $[-T,T]$ provided $l \lesssim 2 W_f T/\pi$.  Therefore, Theorem~\ref{thm:SamplingEngFolk} applies for any sampling rate $W_s \geq W_C +W_f$. 
Applying the interpolation formula to each matrix element (twice) yields \emph{sampled prolate filter diagonalization}:

\begin{definition}[sampled prolate filter diagonalization]
\label{def:sampled_PFD}
    Consider a signal $C(t)$ that is of bandwidth $W_C$, a filter width $W_f$, a sample rate $W_s$ with $W_s\geq W_C+W_f$ and a duration $T$. Assume sample points $C(t_k)$ of the signal 
    at $t_k = \pi k / W_s$ for $k = - N_s , -N_s +1 ,  \cdots, 0, \cdots , N_s -1 , N_s$ where $N_s = 2 \lceil W_s T / \pi \rceil$.\\

    \emph{Sampled prolate filter diagonalization} is the protocol, that generates infinite dimensional matrices $A$ and $B$
    with matrix elements:
\begin{align}
A_{sl} =\;& 
-i \pi^2 \frac{T}{\bws^3 } \sum_{n_1,n_2=0}^{\lfloor 2 W_s T/ \pi \rfloor}  
\mu_{n_2}(c_s) \mu_{n_1}(c_s)  \prlt_{n_1}^{(W_s)}(0)  \prlt_{n_2}^{(W_s)}(0) \notag\\
& \quad \cdot \sum_{k_1,k_2 =-\lceil \bws T / \pi \rceil}^ {\lceil \bws T / \pi \rceil} 
\prlt_{n_2}^{(W_s)}\left(\tfrac{\pi k_2}{\bws}\right)
\prlt_s^{(W_f)}\left(\tfrac{\pi k_2}{\bws}\right)
C\left(\tfrac{\pi k_2}{\bws} - \tfrac{\pi k_1}{\bws} \right)
\prlt_l^{(W_f)\scriptstyle\prime}\left(\tfrac{\pi k_1}{\bws}\right)
\prlt_{n_1}^{(W_s)}\left(\tfrac{\pi k_1}{\bws}\right) \notag\\[1ex]
& + i \pi \sqrt{\frac{T}{W_s^3}} \sum_{n=0}^{\lfloor 2 W_s T/ \pi \rfloor} \mu_n(c_s) \prlt_n^{(W_s)}(0)
\sum_{k =-\lceil \bws T / \pi\rceil }^{\lceil \bws  T / \pi \rceil} \notag\\
& \quad \cdot \prlt_{n}^{(W_s)}\left(\tfrac{\pi k}{\bws}\right) 
\left[
  C\left(\tfrac{\pi k}{\bws} - T \right) {\prlt_l}^{(W_f)}(T)
- C\left(\tfrac{\pi k}{\bws} + T \right) {\prlt_l}^{(W_f)}(-T)
\right]
\label{eq:A_sl}
\end{align}

\begin{align}
B_{sl} =\;& 
\pi^2 \frac{T}{\bws^3} \sum_{n_1,n_2=0}^{\lfloor 2 W_s T / \pi \rfloor}
\mu_{n_2}(c_s) \mu_{n_1}(c_s)
\prlt^{(W_s)}_{n_1}(0) \prlt_{n_1}^{(W_s)}(0) \notag\\
& \quad \cdot \sum_{k_1,k_2 = -\lceil \bws T / \pi \rceil }^{\lceil \bws T / \pi \rceil }
\prlt_{n_2}^{(W_s)}\left(\tfrac{\pi k_2}{\bws}\right)
\prlt_s^{(W_f)}\left(\tfrac{\pi k_2}{\bws}\right)
C\left(\tfrac{\pi k_2}{\bws} - \tfrac{\pi k_1}{\bws} \right)
\prlt_l^{(W_f)}\left(\tfrac{\pi k_1}{\bws}\right)
\prlt_{n_1}^{(W_s)}\left(\tfrac{\pi k_1}{\bws}\right)
\label{eq:B_sl}
\end{align}
Here $c_s= W_s T$ and $\{\prlt_n^{(W_s)} \}$ denotes the sequence of prolates of bandwidth $W_s$ and concentration in the time interval $[-T,T]$.
At a given trial dimension $M$, sampled PFD returns the $M$-th leading principal sub-matrices of $A$ and $B$.
\end{definition}
Sampled PFD is a subspace protocol as in Definition \ref{def:SubspaceProtocol} with 
\begin{align}
    \delta A = A - V^\dagger H V,  \qquad \delta B = B - V^\dagger V, \qquad   \mathcal{N}^{(B)} = N^\dagger N + \delta B.
\end{align}
Here, $V^\dagger H V$ and $ V^\dagger V$ refer to matrices as in PFD (Definition \ref{def:PFD}).
Then, the Master Theorem \ref{thm:subspace_master} yields for sampled PFD the accuracy guarantee:
\begin{theorem}
    \label{thm:sampledPFD}
    Let $[A_M^m, B_M^m]$ be the GEP that is obtained from sampled PFD through Algorithm \ref{alg:eps_dim_red} with a trial dimension $M$ and a detected dimension $m$.
    Assume that the detected dimension $m$ coincides with number of frequencies in $B_{f}$, that the GEP is well-conditioned 
    in the sense 
    \begin{align}
        \lambda_m(B_M^m) > \| \mathcal{N}^{(B)}_M \|,
    \end{align}
    and that the generalized eigenvalues lie within $[-W_f, W_f]$.
    Let $\tilde E_1 > \tilde E_2 > \cdots > \tilde E_m$ denote these eigenvalues, and $E_1 > E_2 > \cdots > E_m$ the frequencies of $C$ within $[-W_f, W_f]$. Then we have
    \begin{align}
   \frac{ \epsilon^{\text{Prlt}}_M \int_{\lambda < -W_f} (\lambda - \tilde E_i) d \alpha(\lambda)  + \lambda_{\min}^*(\delta A_M^m - \tilde E_i \delta B_M^m )}{\lambda_m(B_M^m) - \| \mathcal{N}^{(B)}_M \| } \leq \tilde E_i -  E_i   \qquad \notag \\
   \qquad \leq  \frac{ \epsilon^{\text{Prlt}}_M \int_{\lambda > W_f} (\lambda - \tilde E_i) d \alpha(\lambda) + \lambda_{\max}^*(\delta A_M^m -\tilde E_i \delta B_M^m)  }{\lambda_m(B_M^m) -  \| \mathcal{N}^{(B)}_M \|}.
   \label{eq:sampled_pfd_bounds}
\end{align}
\end{theorem}

The matrices $\delta A$ and $\delta B$ describe the additional error introduced by the prolate sampling formula. By Theorem~\ref{thm:SamplingEngFolk}, this sampling error remains small provided that the sampled functions and their derivatives are highly concentrated in the sampling interval. The functions $ f_l^{(\tau)}$ and their derivatives satisfy this condition as long as the number of used prolate filters is sufficiently lower than the essential dimension $ \tilde c_f = 2 W_f T / \pi$:
\begin{align}
    \|f_{l}^{(\tau)} \|_T^2 &\leq 1-\gamma_{l}(c_f),   \\
     \| f_{l}^{(\tau)\scriptstyle\prime} \|_T^2 &\leq (1-\gamma_{l}(c_f)) \left( 2 C_{\text{extra},l}(W_f, T)^{2}+ 2 W_C^2 \right).
\end{align}
The second inequality follows from Eq.\ \eqref{eq:thm_derivConcen_extra} of Theorem~\ref{thm:prlt_bound} and $|a+b|^2 \leq 2|a|^2 + 2|b|^2$ for $a,b \in \mathbb{C}$.
Inserting these estimates into Theorem~\ref{thm:SamplingEngFolk} shows that the sampling error is controlled by a prefactor $(1-\gamma_{l}(c_f))$. Therefore, the sampling error undergoes a sharp increase as the number of used trial vectors approaches $\tilde c_f$.
The number of trial vectors required depends on the effective spectral density of the signal within the spectral region of interest. In particular,  the analysis of Section \ref{sec:PFD} carries over: in sampled PFD
the discretization error exhibits a phase-transition-like increase, as the effective spectral density approaches $T/\pi$.

In summary, both the subspace-based error as well as the discretization error of our method exhibit a simultaneous transition in accuracy which resembles a generalization of the Heisenberg uncertainty relation. Our analysis relates 
observation time and spectral density to achieved accuracy. In particular, these information-theoretic resources can be balanced to achieve optimal efficiency in computations.

\subsection{Autocorrelation generation on a quantum computer}
\label{autocorrelation_sec}

With sampled PFD established, it remains to generate an autocorrelation function that encodes the spectral information of a physical system.
As discussed in Section~\ref{sec:sig_schrö}, any continuous positive-definite signal can be interpreted as an autocorrelation function arising from unitary time evolution. 
Since simulating unitary time evolution is widely believed to be exponentially hard on a classical computer \cite{feynman_simulating_1982, lloyd1996universal}, we offload the signal generation onto a quantum processor.

Specifically, we use the Hadamard test to estimate the signal samples
\begin{align}
    C(t_k) = \langle  \Psi | e^{i H t_k} | \Psi \rangle = \sum_n |a_n|^2 e^{i E_n t_k} \label{eq:autocorr_tk}
\end{align}
at the equidistant sampling points $t_k = \pi k / W_s$.
Here, $H$ is the self-adjoint operator that describes the system, $\Psi $ is a chosen initial state, and $|a_n|^2 = |\langle \varphi_n, \Psi\rangle|^2$ are the amplitudes of the initial state with respect to the eigenbasis of $H$. 

The Hadamard test separately estimates the real part ($\Re$) or imaginary part ($\Im$) of a signal sample $C(t_k)$ through the following quantum circuit: 
\begin{equation}
\label{scattering-circuit-2}
    \Re/ \Im [ \langle \Psi | U(t_k) | \Psi \rangle ]  = \Qcircuit @C=1em @R=1em {
 & \qw & \gate{H} & \ctrl{1} & \gate{I/S} & \gate{H} & \meter \\
 & {/}^{\otimes^n} \qw & \gate{P} & \multigate{0}{U(t_k)} & \qw & \qw & \qw
}
\end{equation}
Where $S$ is the phase gate and $P$ is the unitary preparing the system register into the state $|\Psi\rangle $.
Choosing $I$ yields a measurement outcome $c_x(t_k) \in \{-1, +1\}$ distributed as
\begin{align}
\mathbb{P}[c_x(t_k) = \pm 1] = \frac{1}{2} \big(1 \pm \Re [C(t_k)]\big),
\end{align}
while using $S$ gives
\begin{align}
\mathbb{P}[c_y(t_k) = \pm 1] = \frac{1}{2} \big(1 \pm \Im[C(t_k)]\big).
\end{align}
In particular,
\begin{align}
\mathbb{E}[c_x(t_k) + i c_y(t_k)] = C(t_k).
\end{align}
Taking the average of $\#_{\text{shots}}$ circuit measurements gives an unbiased estimator for $C(t_k)$:
\begin{align}
    \tilde C(t_k) = C(t_k) + n(t_k),
\end{align}
where $n(t_k)$ describes the statistical noise. 
By Hoeffding's inequality, the number of circuit measurements required to ensure $|n(t_k)| \leq \epsilon$ with probability at least $1 - \delta$ is
\begin{align}
\#_{\text{shots}} = \Omega \left( \frac{1}{\epsilon^2} \log \frac{1}{\delta} \right). 
\label{eq:ShotScalingPoint}
\end{align}
The required number of circuit measurements to ensure  $|n(t_k)|\leq \epsilon$ for all $k=1, \cdots N_s$ with probability of at least $(1-\delta)$ is estimated by a union bound as
\begin{align}
    \#_{\text{shots}} = \Omega \left( \frac{1}{\epsilon^2}\log \frac{N_s}{\delta} \right).  \label{eq:ShotScalingN_s}
\end{align}
In this work, we do not account for simulation error from time propagation on a quantum computer. This error can be made arbitrarily small by increasing circuit depth (or even completely vanish assuming access to a block encoding of $H$ \cite{low2019hamiltonian}).

We term the hybrid procedure consisting of signal-sample generation on a quantum computer followed by sampled PFD as a classical post-processing routine quantum prolate diagonalization (QPD).

\section{Numerical results}

In this section, we present the results of our implementation of QPD. We emulate the signal samples produced by the Hadamard test circuit by modeling shot noise arising from a finite number of circuit measurements (see Section II.F). We first demonstrate the ability of the protocol to recover excitation energies to chemical precision for the lowest $\pi\rightarrow\pi^*$ singlet excitation in benzene from a noisy signal. For this example, we also report results obtained from exact autocorrelation samples.

Subsequently, we study the features of QPD in more detail including shot noise in all experiments. 
We begin by showing that, for the specific task of ground state energy estimation, the protocol exhibits Heisenberg-limited scaling. In this example, we consider the $\mathrm{H}_8$ hydrogen chain as a prototypical strongly correlated system. Next, we move on to the more general task of multiple eigenvalue estimation. 
We anticipate that, in this case as well, the protocol achieves Heisenberg scaling for both weakly (LiH) and strongly ($\mathrm{H}_8$) correlated systems, with total runtime scaling as $T_{\text{runtime}} = \mathcal{O}(\epsilon^{-1})$, and with maximal evolution times shorter than those required by standard QPE, i.e., $T_{\text{max}} \ll \frac{\pi}{\epsilon}$.
Finally, we discuss another important aspect of interference-based quantum algorithms: the effect of the initial input state.

\subsection{Computational details}\label{computational_det}

\subsubsection{Efficient computation of prolates}
Prolates are eigenfunctions of the time- and band-limiting operator $\mathcal{B}_{W}\mathcal{D}_{T}$. However, the direct discretization of the integral equation is ill-conditioned, as the eigenvalues are highly clustered around one and zero.
Instead, we utilize the property---famously termed the ``lucky accident''~\cite{SlepianComment}---that prolates are also eigenfunctions of the prolate spheroidal wave equation. Specifically, $\xi_n^{(W)}(t)$ corresponds to the angular prolate spheroidal wave function $S_{0n}(c, \eta)$ with $c=WT$ and scaled time coordinate $\eta = t/T \in [-1,1]$. The corresponding integral operator eigenvalues $\gamma_n$ (required for the error bounds $\epsilon_M^{Prlt}$) are obtained from the radial prolate functions $R_{0n}^{(1)}(c, 1)$ using the relation $\gamma_n = \frac{2c}{\pi} [R_{0n}^{(1)}(c, 1)]^2$~\cite{ProI}.

We employ the open-source library \texttt{prolate\_swf}~\cite{VanBuren_prolate_swf}, which implements the efficient numerical scheme detailed by Van Buren and Boisvert~\cite{Buren2002AccurateCO}. Their approach overcomes catastrophic cancellation effects in the computation of the radial functions $R_{0n}^{(1)}(c, \zeta)$ that limited previous methods\,\footnote{We denote the radial coordinate as $\zeta$ to distinguish it from the prolate functions $\xi_n$. Standard literature~\cite{Buren2002AccurateCO} uses $\xi$ for this coordinate.}. The angular wave functions are calculated through an expansion in Legendre polynomials:
\begin{equation}
S_{0n}(c, \eta) = {\sum_{k=0,1}^{\infty}} d_k(c|0n) P_k(\eta)
\end{equation}
where the summation runs over even $k$ for even $n$ and odd $k$ for odd $n$. The same coefficients $d_k$ determine the radial functions through a series of spherical Bessel functions $j_k(c\zeta)$, which remains numerically stable even for large $c$ where traditional power series fail. These coefficients satisfy a three-term recurrence relation. The library determines them with Bouwkamp's algorithm~\cite{Bouwkamp1970} and fixes the overall scaling through the normalization scheme of Meixner and Sch\"afke~\cite{meixner_mathieusche_1954} to prevent subtraction errors.
As the prolate generation is part of the classical post-processing routine (sampled PFD), it is carried out entirely on a classical computer.

\subsubsection{Hamiltonian representation and state preparation}

The Jordan--Wigner mapping was employed in all calculations, as implemented in the PennyLane library~\cite{bergholm2018pennylane}. All calculations were performed employing the STO-3G basis set\,\cite{hehre1969quantum}. While we acknowledge that this basis set does not yield highly accurate results, we adopt it for the sake of simplicity, as the primary goal of the following numerical experiments is to evaluate the features of the QPD.
The initial input wave function was prepared and uploaded to the quantum register using a custom modification of the Overlapper library~\cite{Overlapper}, which called PySCF\,\cite{sun2020recent} (see below). All initial approximate wave functions applied in our calculations had the following structure:
\begin{align}
    |\Psi_{input}\rangle = \frac{1}{\sqrt{n_{\rm ts}}} \sum_i^{n_{ts}} |i\rangle
    \label{eq:input_state_expansion}
\end{align}
where $n_{\rm ts}$ is the number of target states of our calculation and $|i\rangle$ is a normalized approximation of the true eigenstate $|i_{\rm FCI}\rangle$ expressed as a linear combination of $K$ Slater determinants $|\Phi_j\rangle$,
\begin{align}
    |i\rangle = \mathcal{N}\sum_j^{K}c_j|\Phi_j\rangle \, ,
    \label{eq:reference_state_expansion}
\end{align}
with normalization constant $\mathcal{N}$.

\subsubsection{Emulation of QPD}

We emulate the Hadamard test circuit through a custom implementation based on JAX~\cite{jax2018github}.
This quantum subroutine of QPD generates signal samples $C(t_k= \pi k / W_s)$ for $k = 1 , \cdots , N_s$. We denote the maximal evolution time as $T_{\text{max}} = \pi N_s / W_s$ and set $T = T_{\text{max}} /2$. 
Exploiting the symmetry $C(-t) = C(t)^*$ gives $2 N_s +1 $ signal samples in $[-2T, 2T]$. Then, sampled PFD can be applied to compute the frequencies of the signal within the target frequency band $B_{f, \omega^*} = [\omega^*- W_f, \omega^* + W_f]$. 
Note that we require a sampling rate large enough such that $W_s \geq W_f + W_C$, where $W_C$ is the signal bandwidth. To obtain a pre-trial for the filter center $\omega^*$ when approximating the lowest lying energies, we used a configuration interaction singles and 
doubles (CISD)\,\cite{sherrill1999configuration} calculation with the PySCF program\,\cite{sun2020recent} (note that alternatives are also possible, such as Ref.\,\cite{vasiliev1999ab}). All reference wave functions of both full configuration 
interaction (FCI, representing an exact diagonalization result) and complete active space configuration interaction (CASCI, equivalent to FCI in a restricted orbital space) quality have been obtained with PySCF\,\cite{sun2020recent}.

As discussed in Section\,\ref{autocorrelation_sec}, finite circuit measurements introduce statistical noise to the signal samples $C(t_k)$. This stochastic noise constitutes an error source distinct from the deterministic subspace and discretization errors,
and is not captured by the approximation theorem in Section\,\ref{sec:sampledPFD}. To evaluate robustness of QPD under realistic conditions, we incorporate this shot noise into our simulations and analyze the average performance as a function of the total quantum runtime $T_{\text{runtime}}$~\cite{yi2024quantum}:
\begin{align}
    T_{\text{runtime}} = \#_{\text{shots}} \sum_{k=1}^{N_s} t_k.
    \label{runtime}
\end{align}
Evaluating Eq.\,\eqref{runtime} with the Gauss formula yields
\begin{align}
    T_{\text{runtime}} & \approx \#_{\text{shots}} \frac{\pi}{W_s}  \frac{N_s^2}{2} = \#_{\text{shots}}  \mathcal{O}\left(  W_s T^2 \right).  \label{eq:T_runtime_scaling}
\end{align}
Furthermore, we characterize the error scaling with respect to the maximum evolution time $T_{\text{max}}$. While our simulations employ exact time evolution, practical implementations on real quantum hardware incur evolution errors (such as Trotterization) that accumulate with $T_{\text{max}}$. Consequently, $T_{\text{max}}$ serves as a direct proxy for the required circuit depth and hardware coherence time. 
Demonstrating chemical accuracy at low $T_{\text{max}}$ is therefore necessary for the algorithm's applicability on early fault-tolerant devices.

For our scans in $T_{\text{max}} \sim N_s$ we set the number of shots per sample to $\text{\#}_{\text{shots}}  = F \sqrt{N_s \log (N_s)}$, where $F$ is some constant independent of $N_s$. We found that an appropriate choice of $F$ depends on the amplitude of the targeted frequencies. 
For weakly correlated systems, where the initial trial with few determinants already has a large overlap with the target state, a value of $F = 2$ was sufficient. For strongly correlated systems---characterized by a low overlap 
(i.e., small amplitudes for the target frequencies) with the target eigenstates and a high number of frequencies in the autocorrelation signal (see Section~\ref{input_state_study})---we set $F = 10$. We will rationalize this behavior 
when discussing the results of Figure\,\ref{mee_heisenberg}.

All numerical results were obtained with the \texttt{Quantum-Prolate-Diagonalization} package, which we developed for the purpose of this paper and make available in the GitHub repository \cite{QPDcode}.

\subsection{Multiple eigenvalue estimation at chemical accuracy}
\label{sec:mee_estimate}

We propose a new algorithm for the simultaneous estimation of multiple eigenvalues of a quantum Hamiltonian $H$. As an example, we focus on the electronic Hamiltonian $H_{el}$,
\begin{equation}
\label{molecular_hamiltonian}
    H_{el} = \sum_{p,q} h_{pq} a^{\dagger}_{p}a_q + \frac{1}{2}\sum_{p,q,r,s} g_{pqrs} a^{\dagger}_{p}a^{\dagger}_{r}a_{s}a_{q}
\end{equation}
where $a_p$ ($a^{\dagger}_p$) is  the annihilation (creation) operator acting on the spin-orbital $p$, and $h_{pq}$ and $g_{pqrs}$ are the usual one- and two-electron integrals\,\cite{helgaker2013molecular}
accounting for the kinetic and electron-nuclei and for the electron-electron interaction terms, respectively.

We demonstrate in the following that QPD can resolve multiple energy levels at chemical accuracy (that is, an error $\epsilon \leq 1\,\text{mHa}$ in the total electronic energy is obtained). As a prototypical test case, we estimate the excitation energy between the two lowest singlet states of benzene, corresponding to a $\pi \rightarrow \pi^*$ transition.
We restricted the calculation to an active space of six electrons in six orbitals, denoted as (6e,6o), which corresponds to the $\pi$-system of benzene.

First, we quantify the subspace error estimates provided by Theorem\,\ref{thm:PFD_accuracy}. 
For this purpose, we do not include shot noise on the signal samples. We prepare an initial state with multiple excitations outside of the band of interest. In more detail, our initial state contains the ground state and the first excited state --- representing 
the in-band frequencies --- and the 6th, 7th, 8th, and 9th excited states as out-of-band frequencies. 
All the eigenstates have the same squared overlap
with the initial state  $|\langle \Psi_{input} | i \rangle|^2 \approx 0.166$. We obtained the excited states $|i \rangle$ using the 35 most relevant configurations obtained from the corresponding CASCI wavefunction (see Eq.\,\eqref{eq:reference_state_expansion}). Table\,\ref{benzene_ee_table} lists all excitation energies of all the states included in the signal with respect to the ground state energy. These frequencies have been shifted by $\omega^* = (E_{GS}' - E_1')/2 = 6.31271$ with respect to the energies $E_i'$ of the CASCI calculation. By construction, the signal bandwidth is $W_C = 0.54191$.

\begin{table}[H]
    \small
    \caption{Energies (E) and excitation energies (EE) of the benzene electronic states included in the initial state used to obtain the results of Fig.\,\ref{benzene_noiseless_conditioning_plot} and Table\,\ref{numerical_details_bounds_results}. Values are reported in atomic units (Hartree) and shifted by the filter center $\omega^*$ and frozen core energy $E_{core} = -221.50532\,\,\text{Hartree}$.}
    \label{benzene_ee_table}
    \centering
    \begin{tabular*}{\textwidth}{@{\extracolsep{\fill}}lrrrrrr}
        \hline
        \hline
        & $E_{|\rm GS\rangle}$ & $E_{|\rm S_1\rangle}$ & $E_{|\rm S_6\rangle}$ &  $E_{|\rm S_7\rangle}$& $E_{|\rm S_8\rangle}$ & $E_{|\rm S_9\rangle}$ \\ \hline
        E & -0.13008 & 0.13008 & 0.29995 & 0.30997 & 0.31004 & 0.41183 \\
        EE & -- & 0.26015 & 0.43003 & 0.44005 & 0.44012 & 0.54191 \\ \hline\hline
    \end{tabular*}
\end{table}

We set for our QPD calculations $W_f = E_{|S_3\rangle} - \omega^* = 0.29170$, $T_{max} = 62.832$ and $W_s = 5$.
Once the signal is sampled, PFD generates a refined GEP according to Algorithm~\ref{alg:eps_dim_red}. The choice of the trial dimension $M$ can be optimized to minimize the subspace error. By Theorem\,\ref{thm:PFD_accuracy} this subspace error is bounded by,
\begin{align}
  \tilde E_i -E_i & \leq \epsilon^{\text{Prlt}}_M \frac{\int_{\lambda > W_f} (\lambda - \tilde E_i) d \alpha(\lambda)  }{\lambda_m(B_M^m) -  \| \mathcal{N}^{(B)}_M \|} \\
   & \leq \epsilon^{\text{Prlt}}_M \frac{\sum_{E_k > W_f} (E_k - \tilde E_i) |a_k|^2  }{\lambda_m(B_M^m) -   \epsilon^{\text{Prlt}}_M },
   \label{eq:sub_err_quantification}
\end{align} 
where, in the last inequality, we exploited Eq.\,\eqref{eq:noise_weigt_bound}. Hence, we neglect the discretization error in our analysis. 
Estimate \eqref{eq:sub_err_quantification} is characterized by three different contributions: (i) an error factor $\epsilon_M^{\text{Prlt}}$, (ii) conditioning $\lambda_m(B_M)^{-1}$, and (iii) integral terms that are specific to a given eigenvalue. As shown in Fig.~\ref{benzene_noiseless_conditioning_plot}, increasing the guess dimension leads to an increase in the subspace error due to the growth of $\epsilon_M^{\text{Prlt}}$. Conversely, the conditioning of the GEP improves as the refinement is performed with a larger number of vectors.

\begin{figure}[H]
    \centering
    \includegraphics[width=0.85\linewidth]{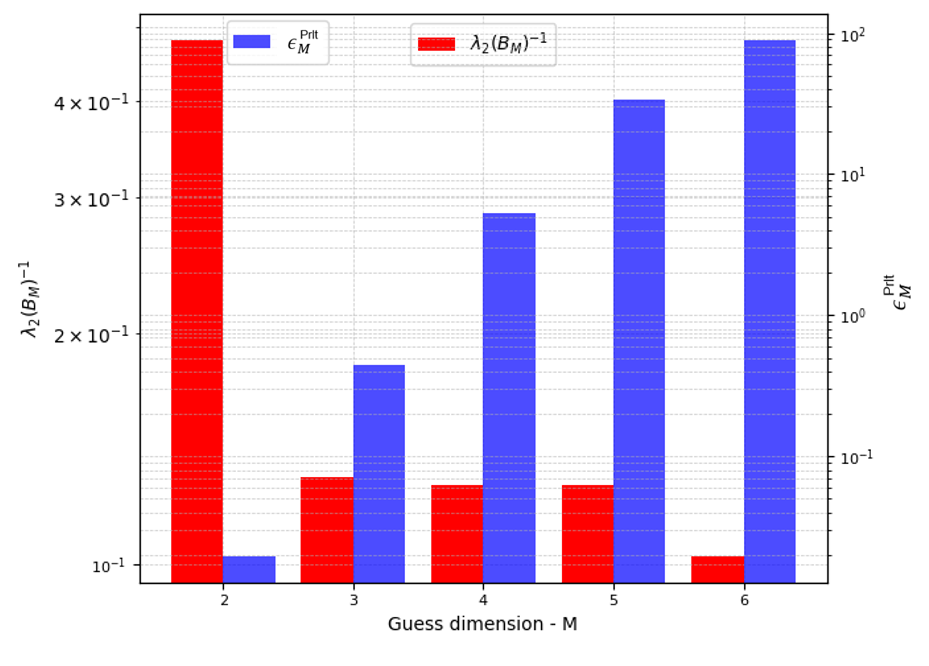}
    \caption{Effect of the guess dimension $M$ on the stability and accuracy of QPD. Conditioning $\lambda_2(B_M)^{-1}$ of the refined GEP (red bars) and the subspace-based error parameter $\epsilon_M^{\text{Prlt}}$ (blue bars) at different trial dimensions $M$. The error parameter is shown on a logarithmic scale.
    }
    \label{benzene_noiseless_conditioning_plot}
\end{figure}

However, Fig.\,\ref{benzene_noiseless_conditioning_plot} shows that the decrease of $\lambda_m(B_M)^{-1}$ in $M$ is not significant. Therefore, we choose $M=2$ in our computations to suppress out-of-band frequencies. 
Table~\ref{error_bounds_table} lists the errors of the energy approximations and the corresponding bounds. The results confirm that Eq.~\eqref{eq:sub_err_quantification} provides sharp estimates of the actual errors of QPD.

\begin{table}[H]
    \small
    \caption{Error quantification of QPD. We report the error obtained numerically, $\epsilon$, and the one estimated by Eq.\,\eqref{eq:sub_err_quantification}, $\epsilon_{bound}$ (given in Hartree atomic units).}
    \label{numerical_details_bounds_results}
    \centering
    \begin{tabular*}{\textwidth}{@{\extracolsep{\fill}}lrr}
        \hline
        \hline
        & $E_{|\rm GS\rangle}$ & $E_{|\rm S_1\rangle}$ \\ \hline
        $\epsilon$ & 0.000924 & 0.001077 \\ 
        $\epsilon_{bound}$& 0.002931 & 0.001286 \\\hline\hline
    \end{tabular*}
    \label{error_bounds_table}
\end{table}

We now move to our second experiment where we include statistical noise due to finite measurements on the signal samples and demonstrate that QPD, nevertheless, resolves multiple energy estimates at chemical accuracy. In this example, we construct the input state from the Hartree–Fock determinant for the ground state and the five dominant configurations of the first excited singlet, 
including the $\text{HOMO}-1 \rightarrow \text{LUMO}$ and $\text{HOMO} \rightarrow \text{LUMO}+1$ transitions.
Owing to the weakly correlated character of the system, this initial state suffices to generate an autocorrelation function dominated by the ground and first excited state.

For this calculation, we set $\omega^* = 6.3$, $W_f = 1$, $W_s = 3$, and $T_{\text{max}} = 104.72$ (in Hartree atomic units, a.u.). The autocorrelation samples $C(t_k)$ were estimated using only $\#_{\text{shots}} = 13$ circuit shots per signal sample, 
yielding particularly noisy input data $\tilde C(t_k)$.
The results following the workflow of PFD are presented in Figure\,\ref{benzene_calc}. 
The spectra of the weight matrices $B_M$ confirm that the signal has $m=2$ dominant frequencies in $B_{f,\omega^*}$ (Figure\,\ref{benzene_calc}a).
We set the trial dimension to half of the critical dimension $M = \lfloor W_f T / \pi \rfloor = 16$. Algorithm\,\ref{alg:eps_dim_red} generates a refined GEP $[A_M^m, B_M^m]$ whose eigenvalues approximate the in-band frequencies. The corresponding amplitudes are estimated through 
Eq.\,\eqref{eq:Amp_estimation}. Figure\,\ref{benzene_calc}c--d depict the reconstruction of the filtered signal (blue solid line) from the noisy samples (red stars). Note that we actually depict the shifted autocorrelation 
$C(t)e^{-i \omega^* t}$ such that relevant energies appear in the low-frequency range. The energy estimates of our computation and their accuracy with respect to a CASCI reference calculation are listed in Table~\ref{numerical_details}. In particular, we find that QPD surpasses chemical accuracy.

\begin{figure}[h!]
    \centering
    \includegraphics[width=\linewidth]{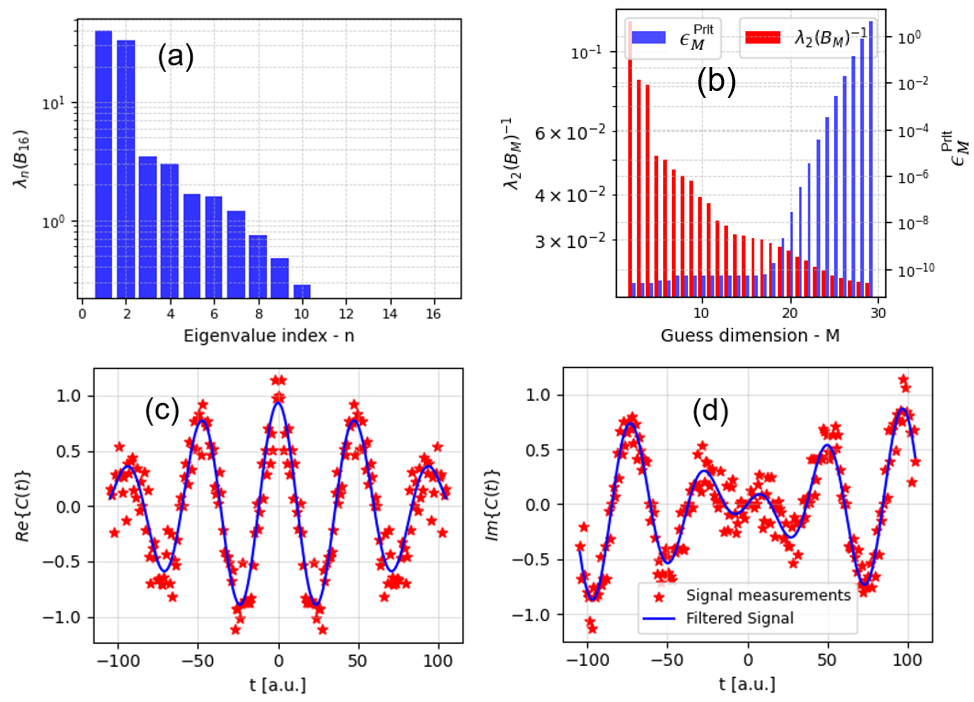}
    \caption{QPD applied to the benzene molecule. (a) Noisy samples of the electronic autocorrelation function were applied to generate the weight matrix $B_{16}$ (Eq.\,\ref{eq:B_sl}) at trial dimension $M=16$. Its spectrum was exploited to detect the number of dominant frequencies $m$ within $B_{f, \omega^*}$. 
    (c) and (d) Red stars denote noisy samples of the shifted signal $\tilde C(t_k)e^{-i\omega^* t}$; the blue solid line shows the reconstructed filtered signal. (c) shows the real part, and (d) the imaginary part.
    (b) Conditioning $\lambda_2(B_M)^{-1}$ of the refined GEP (red bars) and the subspace-based error parameter $\epsilon_M^{\text{Prlt}}$ (blue bars) at different trial dimensions $M$.The error parameter is shown on a logarithmic scale.
    }
    \label{benzene_calc}
\end{figure}

\begin{table}[H]
    \small
    \caption{Absolute energies of the ground (GS) and first excited (S$_1$) states (middle columns) and the lowest $\pi\rightarrow\pi^*$ singlet energy gap (last column) of the benzene molecule. Results obtained with QPD and with CASCI are reported in Hartree.}
    \label{numerical_details}
    \centering
    \begin{tabular*}{\textwidth}{@{\extracolsep{\fill}}lrrr}
        \hline
        \hline
        & $E_{|\rm GS\rangle}$ & $E_{|\rm S_1\rangle}$ & $\Delta E_{|\rm GS\rangle \rightarrow |\rm S_1\rangle} $ \\ \hline
        QPD (6e,6o) & -227.948210 & -227.688626 & 0.259584 \\ 
        CASCI (6e,6o)& -227.948115 & -227.687960 & 0.260154 \\
        $\epsilon$ & 0.000095 & 0.000666 & 0.000570 \\\hline\hline
    \end{tabular*}
\end{table}

Our choice to set the trial dimension $M$ to half the essential dimension is not necessarily optimal. However,  we found that this setting consistently achieves low errors due to 
off-band frequencies and prolate sampling, while maintaining good conditioning of the weight matrix $B_M^m$. For noisy input data $\tilde C(t_k)$, it is especially beneficial to choose a larger trial dimension $M \gtrsim m$, as it improves the conditioning $\lambda_m(B_M^m)^{-1}$ (see Appendix~\ref{sec:SubspaceProtocols} and \cite[Section~4]{stroschein2025approximationframeworksubspacebasedmethods}). Better conditioning, in turn, improves the stability of the method with respect to additional error sources (Theorem~\ref{thm:subspace_master}). 

Regarding this aspect, the stability achieved by QPD is quite remarkable, as we are able to estimate the two eigenvalues within chemical accuracy using only 13 shots per signal sample. We compare this result with other examples in the literature~\cite{klymko2022real, lee2024sampling, shen2025estimating}.

We begin with the in-depth analysis presented in Ref.~\cite{lee2024sampling}, where the authors prove that, in order to maintain a constant conditioning error, the number of measurements per matrix element of a GEP must scale as $\mathcal{O}(n^3)$, with $n$ denoting the dimension of the GEP. For comparison, all numerical experiments reported in their work require a number of shots on the order of at least $10^6$ measurements per sample. Refs.~\cite{klymko2022real, shen2025estimating} propose an approach very similar to the one considered here, in which a subspace problem constructed by sampling snapshots of the autocorrelation function is used to determine a subset of the eigenvalues of a Hamiltonian $H$. In both cases, to mitigate the ill-conditioning of the noisy GEP, the authors adopt a thresholding strategy\,\footnote{When using singular value thresholding, the GEP is regularized truncating the subspace with a threshold parameter $\delta$. The original GEP $[A, B]$ is then transformed into $[\tilde{A}, \tilde{B}]$ where $\tilde{B} = \sum_{l, \sigma_l \geq \delta} \sigma_l \textbf{u}_l\textbf{v}_l^{\dagger}$ and $\tilde{A}= \tilde{\textbf{U}}^{\dagger} A \tilde{\textbf{V}}$ where $\textbf{U, V}$  are the matrices of the left $\textbf{u}$ and right $\textbf{v}$ singular vectors after truncation (e.g. $\tilde{\textbf{V}}_{ij} = \textbf{V}_{ij} \,, \sigma_j \geq \delta $).}. Their results are obtained by adding Gaussian noise $\mathcal{N}(0,10^{-2})$, which approximately corresponds to a number of measurements on the order of $10^4$.
By contrast, our results demonstrate that our Algorithm~1 is able to reduce the required number of measurements by several orders of magnitude with respect to the results of \cite{klymko2022real, lee2024sampling, shen2025estimating}. Given the importance of this aspect for the entire class of quantum subspace algorithms~\cite{lee2025efficient}, we will provide a more detailed discussion of this topic in future work.

We recall that sampled PFD requires a sufficiently large sampling rate $W_s\geq W_f + W_C$, where $W_C$ is the bandwidth of the shifted signal $C(t)e^{-i\omega^* t}$. A simple way to rigorously estimate $W_C$ is provided by spectral bounds on the electronic Hamiltonian as shown in Ref.\,\cite{cortes2024assessing}. As this probably overestimates the resources needed, a possible alternative might be a cheap classical calculation to obtain a more informed trial of $W_C$. 

Finally, we conclude this section remarking that, as stated in Eq.\,\eqref{eq:accuracy_condition}, the frequency density is a critical parameter at fixed $T_{\max}$. Of course, in routine applications the spectral density of the signal is not known $\textit{a priori}$, but we note that cheap and quick classical calculations\cite{stein2016automated} can provide a useful guide to not overestimate the minimum $T_{\text{max}}$ needed.

\subsection{Numerical evidence for heisenberg-limited scaling}
\label{sec:heisenberg_limit}

In this section, we provide numerical evidence that QPD achieves Heisenberg scaling in both the weak and strong correlation limit for ground state energy estimation and for multiple eigenvalue estimation. This has been previously achieved by only a few other algorithms relying on a single-ancilla phase estimation circuit\,\cite{ding2024quantum, ding2023simultaneous}.
We recall that Heisenberg scaling is the optimal bound achievable according to quantum information theory\,\cite{giovannetti2006quantum}. From an algorithmic point of view, it implies that the quantum computational resources employed are inversely proportional to the error estimate of the protocol, that is, $T_{\text{runtime}} = \mathcal{O}(\epsilon^{-1})$. Opposed to this regime, we have the shot noise limit, where the accuracy of our eigenvalue estimate is bounded by the accuracy of the single sample estimate (see Section\,\ref{autocorrelation_sec}), for which we have $T_{\text{runtime}} = \mathcal{O}(\epsilon^{-2})$.
Moreover, we also monitored how the error estimate depends on $T_{\text{max}}$, which is another important parameter for evaluating the performance of QPD. 
Algorithms with a better accuracy scaling in $T_{\text{max}}$ are particularly valuable for quantum circuits, because they lead to shallower circuits and shorter run times (see Eq.~\eqref{eq:T_runtime_scaling}).

\subsubsection{Ground state energy estimation}
\label{gsee}

We consider ground state energy estimation for the linear, $D_{\infty h}$-symmetric $H_8$ system 
with internucelar distance $r_{\rm HH} = 2 \text{ \angstrom }$. Although experimentally unstable\,\cite{stella2011strong}, hydrogen chains have been extensively characterized as the simplest systems revealing strong electronic correlation phenomena\,\cite{motta2017towards, motta2020ground}. 
Besides this system having a strong multireference character, 
it also features a dense spectrum in the low-energy section. 

Although dense spectra are generally challenging to describe, our results demonstrate that QPD is capable of efficiently and reliably addressing them. The initial input state used for these calculations was obtained by uploading the first five most 
relevant determinants from the exact FCI wave function. Of course, in an actual application where QPD is applied to larger systems (specifically, too larger orbital bases) no FCI reference (i.e., reference from exact diagonalization) wave function will be available. Instead, one may exploit approximate CI methods or approximate matrix product state calculations. The latter have been successfully exploited for the automated construction of active orbital spaces \cite{stein2016automated} and for the analysis and preparation of initial states for QPE \cite{morchen2024classification}.
With such an input state we obtain an initial overlap $|\langle \Psi | GS_{\text{FCI}}\rangle | ^2 \approx 0.33$. We set the sampling rate to $W_s = 3$ \,\,a.u.

\begin{figure}[h!]
    \centering
    \includegraphics[width=\linewidth]{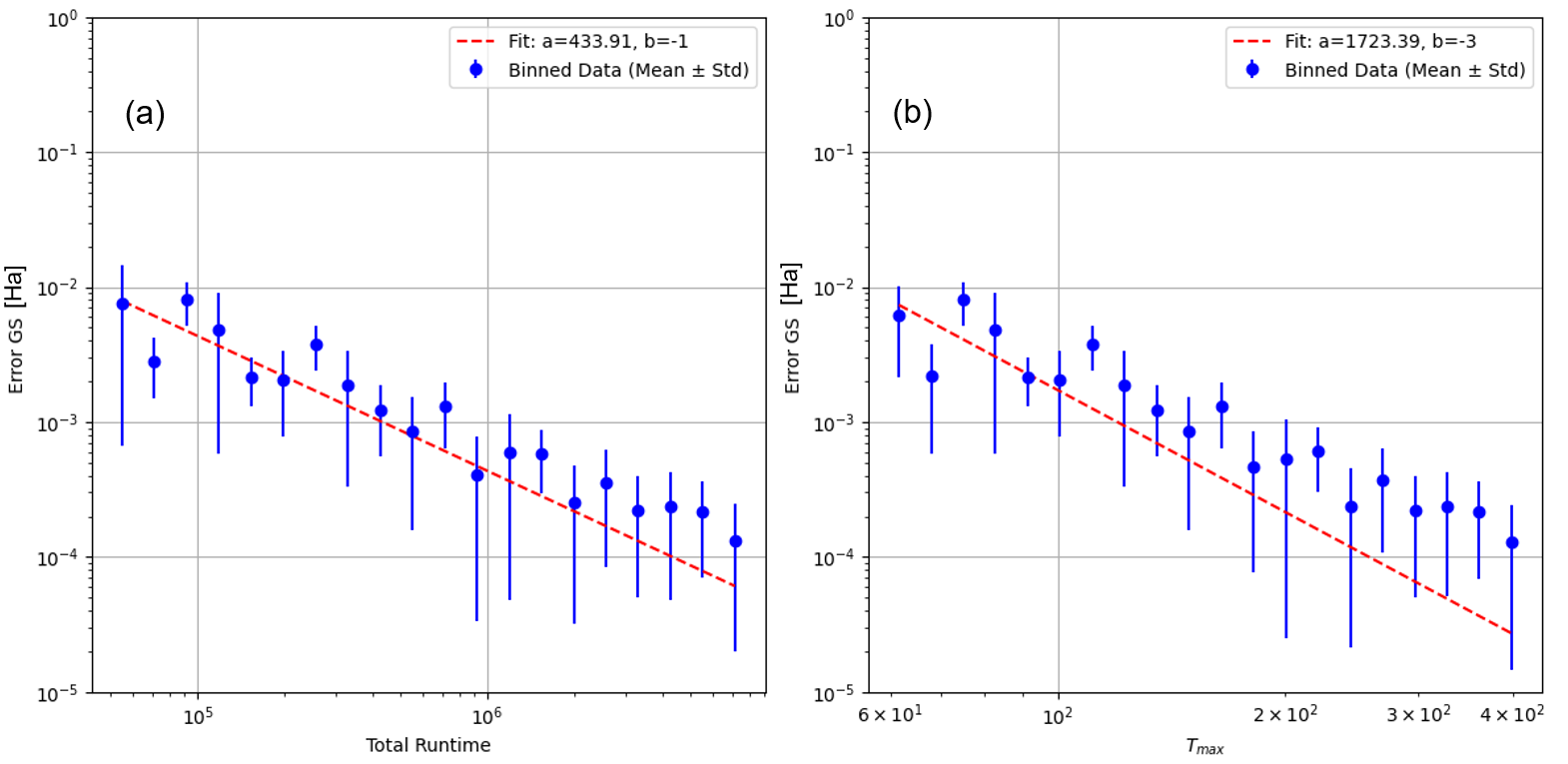}
    \caption{Empirical scaling of QPD 
    for the task of ground state energy estimation for H$_8$.      
    Left (a): average error on the ground state energy as a function of the total runtime (see Eq.\,\eqref{runtime}). 
    Right (b): average error on the ground state energy as a function of the maximal evolution time $T_{\text{max}}$. Blue dots and corresponding error bars are obtained by averaging different runs of QPD where we scan over the maximal evolution time $T_{\text{max}}$. All other parameters have been set according to Section\,\ref{computational_det}. Red dashed lines correspond to the fit of the function $y = a \, e^{b}$; $b$ values are kept fixed and set to $b=-1 \, (-3)$ in the left (right) plot.}
    \label{heisenb-scaling1}
\end{figure}
Figure~\ref{heisenb-scaling1} collects the results of multiple calculations obtained by scanning the total length of the sampled signal. Final results were computed by binning the runtime axis (or the $T_{\text{max}}$ axis) into 20 bins and calculating the mean and standard deviation within each bin. 

In Figure\,\ref{heisenb-scaling1}a, which reports the error on the ground state energy estimate (in Hartree) as a function of the runtime $T_{\text{runtime}}$ (Eq.\,\eqref{runtime}), we notice a very good agreement with the inverse scaling law dictated by the Heisenberg limit.
Figure\,\ref{heisenb-scaling1}b shows that QPD exhibits a scaling of $\epsilon = \mathcal{O}(T_{\text{max}}^{-3})$. 

This result represents a cubic speedup compared to the standard QPE scaling which is of $\epsilon = \mathcal{O}(T_{\text{max}}^{-1})$ and a quadratic improvement with respect to other state-of-the-art approaches \cite{ding2024esprit, ding2024quantum}. 
We point out that Eq.\,\eqref{eq:T_runtime_scaling} and our choice of $\#_{\text{shots}}$ imply a scaling $T_{\text{runtime}} = \mathcal{O}(T_{\text{max}}^{2.5}\sqrt{\text{log}(T_{\text{max}})})$, which is almost cubical and in agreement with our numerical results.
Finally, we note that a cubic improvement in circuit depth (maximal evolution time) has already been shown in the literature when a cosine window was inserted into the standard QPE circuit\cite{rendon2022effects}.

\subsubsection{Multiple eigenvalue estimation}
\label{mee}

We now consider the more general task of multiple eigenvalue estimation. Thus, we aim to determine more than one frequency within a given band $B_{f, \omega^*}$.
We will examine two systems representing opposite limits of weak and strong correlation. As an example for the weak correlation case, we study the LiH molecule at its equilibrium internuclear distance; for the strong correlation regime we again exploit the H$_8$ system. In the first case, all frequencies present in the signal also lie in band $B_{f, \omega^*}$. Therefore, there are no out-of-band contributions. By contrast, the second example corresponds to signals with many frequencies outside of $B_{f, \omega^*}$. In this case, we only recover a small fraction of all the frequencies present in $C(t)$.

Although quantum computing approaches are typically considered to offer the greatest advantages for strongly correlated systems\cite{motta2022emerging}, we believe that a utility-scale quantum algorithm should also be able to address efficiently, that is, with provable advantages over classical computations, a wider span of the electron-correlation (and hence, chemical) space. While electronic correlation has been thoroughly classified~\cite{morchen2024classification,cioslowski2012robust,materia2024quantum}, we here focus specifically on its impact on correlation functions $C(t)$. In Ref.~\cite{zurek1982environment}, Zurek showed that correlation function decay rates and fluctuation amplitudes depend on the number of dynamically active states. The higher the number of states involved in the dynamics the faster the decay of autocorrelations. This is consistent with the orthogonality catastrophe~\cite{anderson1967infrared}: the system's evolution tends towards states with exponentially shrinking overlaps with a given eigenstate. Increasing the number of states also suppresses long-time fluctuations. In weakly correlated systems, $C(t)$ preserves long-time oscillations, even with few initial determinants in the initial input state. By contrast, for strong correlations, we observe rapid damping and suppression of recurrence beatings. Noise sensitivity increases as oscillations fade, masking signal information. 

\begin{figure}[h!]
    \centering
    \includegraphics[width=\linewidth]{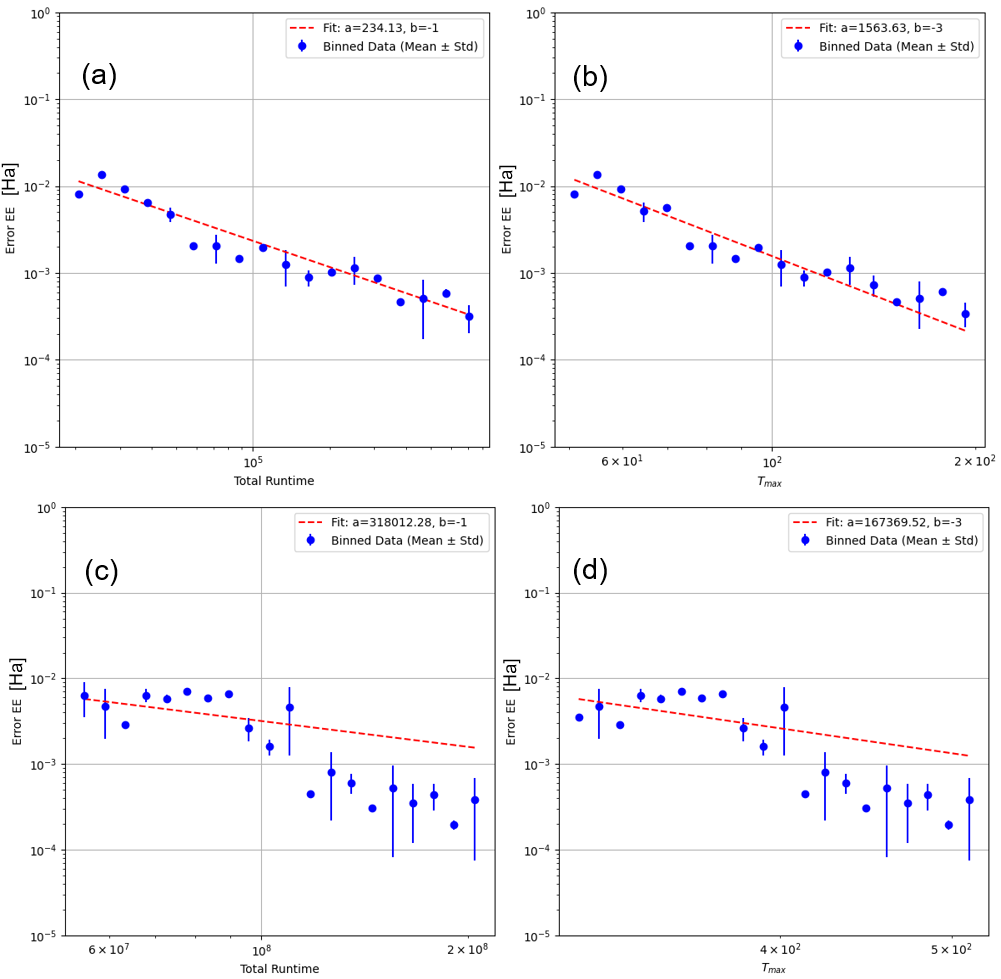}
    \caption{Empirical scaling of QPD   
    in the weak ($LiH$ molecule, a-b) and strong (H$_8$ system, c-d) correlation limit for the task of simultaneous multiple eigenvalue estimation.  
    Left: average error on the excitation energies as a function of the total runtime (see Eq.\,\eqref{runtime}). Right: average error on the excitation energies as a function of the maximal evolution time $T_{\max}$. Blue dots and corresponding error bars were obtained by averaging different runs of QPD 
    where we scan over the maximal evolution time $T_{\max}$.     
    All other parameters have been set according to Section\,\ref{computational_det}. Red dashed lines correspond to the fit of the function $y = a \, e^{b}$; $b$ values were kept fixed and set to $b=-1 \, (-3)$ in the left (right) plot.}
    \label{mee_heisenberg}
\end{figure}

Analogously to the previous section, Figure\ \ref{mee_heisenberg} collects the results of multiple calculations obtained by scanning the total length of the sampled signal. Final results were computed by binning the runtime axis (or the $T_{\text{max}}$ axis) into 20 bins and calculating the mean and standard deviation within each bin. 
For the weak correlation example, we estimated simultaneously the eigenvalues of the four lowest eigenstates of LiH at equilibrium distance ($r_{\rm LiH}=1.6\,\, \text{\angstrom} $). These are: $|\rm GS\rangle, |\rm S_1\rangle, |\rm S_2\rangle$, and $|\rm S_{4}\rangle$. We built the input state by selecting, for each eigenstate, the four most important determinants. Again, we point out that scaling up the initial state preparation problem to larger systems can be done via approximate matrix product states calculations\,\cite{stein2016automated, morchen2024classification}. 
Because of the weakly correlated nature of the system, the overlap of this few-determinant approximation with the target state is about 1.0. In all LiH calculations, we set a value of $W_s = 3$\,\,a.u. As we can understand from the results in Figure \ref{mee_heisenberg}, all considerations and observations made for the ground state energy estimation problem in the previous section also apply to this case, with the average error on the excitation energies being inversely proportional to the quantum runtime and scaling as $\mathcal{O}(T_{\text{max}}^{-3})$ with the maximal evolution time (see Figure~\ref{mee_heisenberg}a–b).
Notably, the multiplicative prefactor (that is, the $a$ parameter in the legend of Figure\,\ref{mee_heisenberg}a-b) resulting from the power-law fit (red dashed line) is slightly smaller than the ones obtained in the previous section (the $a$ parameter in Figure\,\ref{heisenb-scaling1}). 
This observation may be related to the facts that (i) the overall input wave function shows a high fidelity with the target state and that (ii) the energy spectrum is much less dense for LiH than for H$_8$. 
For the multiple eigenvalue estimation on H$_8$ (Figure\,\ref{mee_heisenberg}c-d), the choice of the initial input state requires more attention. Figure~\ref{initial_input_state_appendix} presents a characterization of the input state used to obtain the results of Figure\,\ref{mee_heisenberg}c-d. The initial input state is defined as noted above in Eq.\,\eqref{eq:input_state_expansion} and thus reads 
\begin{align}
    |\Psi\rangle = \frac{1}{2}
\left[
|\rm GS\rangle + |\rm S_1\rangle + |\rm S_2\rangle + |\rm S_{4}\rangle
\right].
\end{align}
The left panel of Figure~\ref{initial_input_state_appendix}a presents the total squared overlap $|\langle \Psi | \Psi_{\text{FCI}} \rangle|^2 = \frac{1}{4}\sum_i|\langle \Psi |i_{\rm FCI} \rangle|^2$ of the state $|\Psi\rangle$ and each corresponding FCI reference eigenstate wave function 
as a function of the number of determinants included in each approximation $|i\rangle$. With increasing number of determinants, the total squared overlap improves monotonically, approaching unity, which shows the systematic convergence of the prepared state toward the target state. 
After having characterized the total squared overlap as a function of the number of determinants included in our approximation for the target eigenstates, we selected for the calculations shown in Figure\,\ref{mee_heisenberg}c-d the 90 most relevant determinants per eigenstate. 

With this setting, the total squared overlap between the initial input state and the target eigenstates is approximately 0.55. 
This value implies for the excited states (whose wavefunction is less peaked than the ground state wavefunction, Fig.\,\ref{initial_input_state_appendix}) a squared overlap $|\langle \Psi|i_{\text{FCI}}\rangle|^2 \approx 0.1$. In all calculations on H$_8$, we set $W_s = 6\,\,\text{a.u.} \,, W_f \approx 0.47 $ and the number of frequencies within $B_{f, \omega^*}$ was set to $m = 7$.

\begin{figure}[h!]
    \centering
    \includegraphics[width=\linewidth]{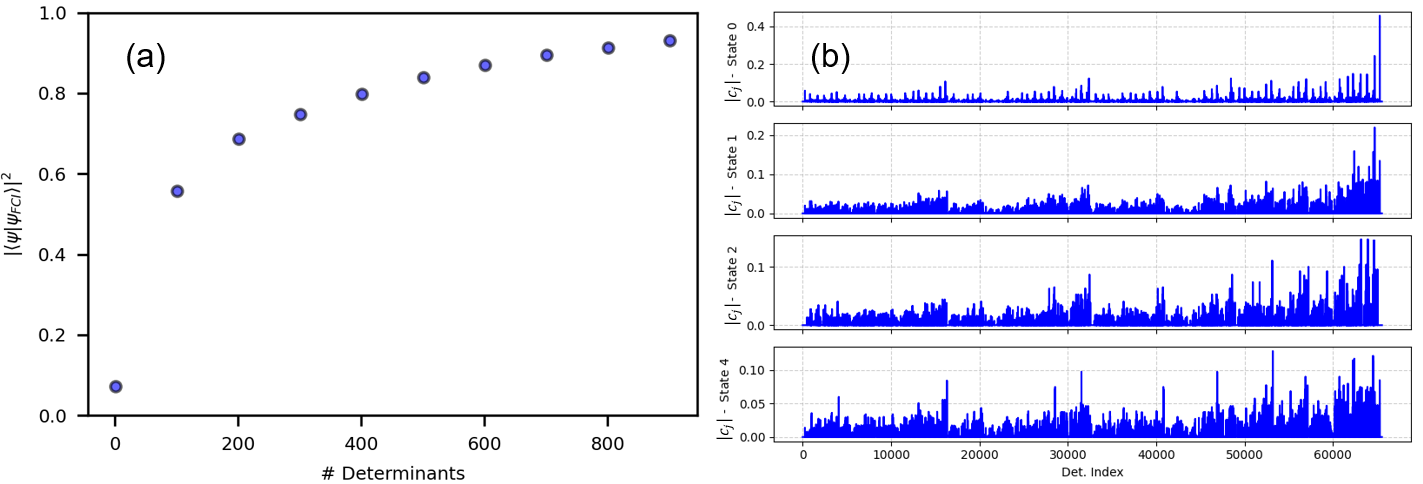}
    \caption{Initial input state characterization for the calculations on H$_8$ reported in Figure\ \ref{mee_heisenberg}c-d. (a) Initial input state total squared overlap as a function of the number of Slater determinants considered. (b) Absolute values of the expansion coefficients $c_j$ 
    for the targeted states: $|GS\rangle$, $|S_1\rangle$, $|S_2\rangle$, and $|S_4\rangle$ (see Eq.\,\eqref{eq:reference_state_expansion}).}
    \label{initial_input_state_appendix}
\end{figure}

For the sake of simplicity, we chose to include the same number of electronic configurations for each target eigenstate. This choice is not optimal and implies that the frequencies we aim to recover are not equally represented in $C(t)$. The reason for this can be seen in Figure~\ref{initial_input_state_appendix}b, which shows the distribution of the absolute value of the expansion coefficients, $|c_j|$, (cf.\ Eq.\,\eqref{eq:reference_state_expansion}) for each determinant contributing to the target eigenstates. While the ground state still shows a slight predominance of the $|\rm HF\rangle$ determinant, we observe that the higher the energy of a state is, the less compact its representation becomes. This is not only due to the intrinsic electronic structure of the system, but also reflects the nature of the CI expansion built upon the mean-field determinant $|\rm HF\rangle$, with orbitals optimized for the ground state by virtue of the variational principle.

We also observe that the average error in the excitation energies (plotted as a function of both the total runtime and $T_{\max}$) follows two distinct regimes, with a transition occurring around $T_{\text{runtime}} \approx 10^8$ and $T_{\max} \approx 400\,\,\text{a.u.}$ This behavior can be explained in the light of Eq.\,\eqref{eq:accuracy_condition}: data points at lower values of $T_{\text{runtime}}$ (and $T_{\max}$) do not satisfy this condition. This is due to the high density of the spectrum of H$_8$, where the eigenvalues of interest are very close in energy. As a result, the accuracy exhibits a marked change in behavior once $T_{\max}$ (and therefore also the total runtime) exceeds a certain threshold.
If we, however, focus on the region $T_{\text{runtime}}\geq 10^8$ and $T_{\max} \geq 400$ \,\,a.u. we find again that the average error on the excitation energies behaves as per Fig.\,\ref{mee_heisenberg}a-b and Fig.\,\ref{heisenb-scaling1}.

Finally, to conclude this section, we emphasize that the spectral density relevant to the accuracy of the sampled PFD is dictated by the choice of $B_{f, \omega^*}$, since QPD is tasked with recovering all frequencies within this bandwidth. For this reason, we study the effect of the initial input state overlap in the next section, which can deeply affect the frequency distribution determining $C(t)$.

\subsection{Effect of the initial input state}
\label{input_state_study}

Selecting and initializing well-suited quantum states is deeply intertwined with the computational difficulty of solving electronic structure problems. In the general case, computing ground state energies of local Hamiltonians is proven to be QMA-hard \cite{kempe2006complexity, o2022intractability}. Nonetheless, from a total resource perspective, encoding a classically obtained state on a quantum device is generally far less resource-intensive than executing a complete quantum energy estimation algorithm, regardless of the specific technique employed \cite{fomichev2024initial, kiss2025early}. In more detail, it has been shown\,\cite{fomichev2024initial} that preparing an initial input state consisting of about a thousand determinants (e.g., expansion terms in Eq.\, (57)) for a system with one hundred orbitals would require $10^5-10^6$ Toffoli gates, while the total circuit cost for a standard QPE circuit for such a system would be on the order of $10^{10}$ Toffoli gates. This means that the cost of preparing an initial input state with sufficient support to run efficiently a phase estimation procedure should be considered as relatively negligible with respect to the gains in terms of total runtime as compared to a poor initialization. Furthermore, if one considers embedding methods\,\cite{erakovic2025high} or recent results regarding MPS preparation\,\cite{berry2024rapid, smith2024constant}, 
this cost is reduced even further.

This fact, along with the consequences of the orthogonality catastrophe \cite{anderson1967infrared}, emphasizes the need to understand how the initial state selection affects the efficiency and accuracy of a quantum algorithm. Assessing the utility of a quantum algorithm hinges on the capacity of the latter to extract eigenvalues at chemically relevant accuracy despite an initial state that can be generated with limited amount of classical resources.
In Figure\,\ref{initial_input_state}, we report the results of the effect of the initial input state quality on the accuracy of the energy estimate of QPD. All calculations were carried out with a sampling rate equal to $W_s = 6$\,\,a.u. For each number of determinants included in the initial input state (ranging from 1 up to 901 determinants), 
we ran 60 QPD calculations scanning the value of $T_{\max}$ from $T_{\max} = 400$ \,\,a.u. to $T_{\max} = 1000$ \,\,a.u. In all linear combinations constructed for the different initial input states, we selected the determinants with the highest weight in the exact FCI reference wave function.

In Figure\,\ref{initial_input_state}, we report the results of in total 1140 different calculations with average results per each fixed number of determinants included in the initial input state (dots) and the corresponding standard deviation (bars).

\begin{figure}[h!]
    \centering
    \includegraphics[width=\linewidth]{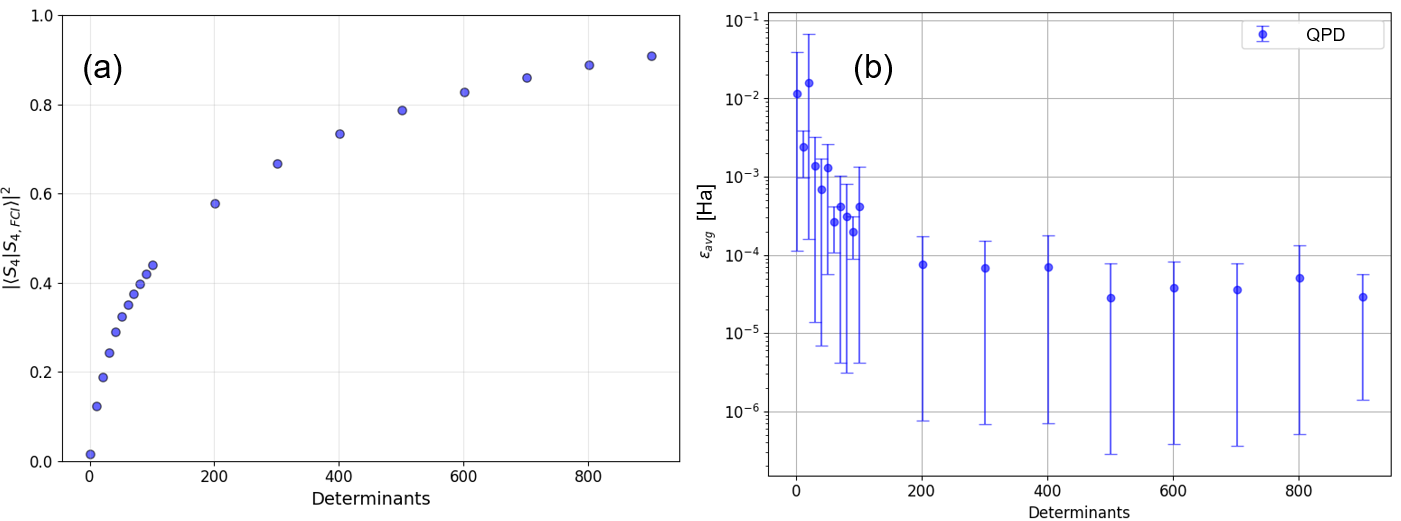}
    \caption{Effect of the initial input state on the QPD algorithm for the H$_8$ system. (a) Initial input state squared overlap as a function of the number of Slater determinants included in the initial input state for the $|S_4\rangle$ state. (b) Average error on the energy estimate as a function of the number of Slater determinants included in the initial input state. Averages were taken over different QPD runs spanning a range of $T_{\text{max}}$ for the sampled autocorrelation function.}
    \label{initial_input_state}
\end{figure}

Figure \ref{initial_input_state}a shows the squared overlap of an approximate state and the exact FCI state as a function of the number of determinants employed in the approximation.  Figure\,\ref{initial_input_state}b shows the average error $\epsilon_{avg}$ (over the multiple runs) for the fourth excited singlet state $|S_4\rangle$ of the $\text{H}_{8}$ system, again as a function of the number of determinants. A logarithmic scale is used on the vertical axis to better visualize the wide range of errors.

In Figure\,\ref{initial_input_state}b, we observe three distinct regimes in the error scaling with respect to the initial state quality. First, the error decreases rapidly until approximately 60 determinants are included, corresponding to an initial squared overlap of $|\langle \Psi|S_{4,\rm FCI}\rangle|^2 \approx 0.35$. In the second regime, adding further determinants continues to reduce the error below the chemical accuracy threshold ($\epsilon_{avg} \leq 10^{-3}$ Hartree), although the rate of improvement is reduced. Finally, beyond 200 determinants, the average error plateaus in the range of $10^{-4} - 10^{-5}$ Hartree.

Our analytical framework (particularly Section \ref{sec:Amp}) provides a rigorous explanation for this behavior. In the first regime, the initial state reaches a sufficient overlap such that the in-band amplitude $|a|_{\text{min}}=|\langle \Psi|S_{4,\rm FCI}\rangle|$ satisfies the detectability condition of Eq.\,\eqref{eq:weight_lowerbound}. Below this threshold, the signal is indistinguishable from the noise background. Beyond this threshold (second regime), increasing the overlap further improves the conditioning of the weight matrix, $\lambda_m(B_M^m)^{-1}$. As established in the general Theorem\,\ref{thm:subspace_master}, a better-conditioned weight matrix linearly suppresses perturbations from arbitrary error sources --- including stochastic shot noise --- leading to the observed accuracy gains.
However, in this regime, Eq.\,\eqref{eq:weight_lowerbound} implies $\lambda_m(B_M^m)^{-1} \sim |a|_{\text{min}}^{-2}$, and the advantage of increasing $|a|_{\text{min}}$ eventually saturates (third regime). Beyond sufficiently large overlaps ($|a|_{\text{min}}^2 \gtrsim  0.5$), further increases yield diminishing returns, as QPD already efficiently suppresses out-of-band frequencies. Consequently, the overall error plateaus at the dominant error level set by the shot noise.

These results demonstrate the robustness of QPD to imperfect state preparation: the method achieves chemical accuracy with initial states accessible via classical approximations, provided they satisfy a modest overlap threshold. This threshold depends on the magnitude of the other error sources.

\section{Conclusions}
We presented a novel eigenvalue problem that makes the duality between spectral analysis, as in eigen decompositions of operators, and frequency decompositions of functions explicit.
To solve this spectral equation, we developed a general subspace-based approximation framework that enables a non-asymptotic error analysis of multiple error sources. We rigorously studied the constraint imposed by
finite observation/simulation time (deriving a non-asymptotic error bound) and the effect of a finite number of signal samples. As a result, we obtained a method that extracts  the frequencies of an autocorrelation function $C(t)$ with high accuracy.
Additional error sources, such as shot noise, can be incorporated seamlessly into the same framework by quantifying their effect on the relevant matrices.

By starting from such a non-asymptotic error framework, our work deviates from previous approaches to rigorous error analysis, such as those in Refs.\,\cite{ding2023even,ding2023simultaneous,li2023adaptive,stroeks2022spectral,yi2024quantum,gunther2025phase,dutkiewicz2022heisenberg,dutkiewicz2024error, ding2024quantum}. 
The framework makes it possible to study multiple error sources in a non-asymptotic manner, which is necessary for error-balancing. This
paves the way toward optimal approximation that connects phase estimation with subspace-based approximation and draws on insights from prolate Fourier theory.

Our analysis yields a quantitative relationship between observation time, spectral density, and achievable precision.
In particular, we derived a sharp \emph{accuracy transition} that resembles the Heisenberg uncertainty relation (Sections~\ref{sec:PFD}-\ref{sec:sampledPFD}). Our accuracy guarantees distinguish between subspace-based error and discretization-based error. Remarkably, the transition appears simultaneously in both, which underscores the maturity of our protocol.

To provide a self-contained and complete picture of our theory, we deferred some of its key parts to the appendix in favor of a streamlined presentation in the main article. In particular,
our theory is built upon the framework of subspace protocols (Section~\ref{sec:SubspaceProtocols}) and insights on optimal time and band concentration (Section~\ref{sec:prolates}). We derived new identities on prolates to give error bounds for our routines (Section~\ref{sec:newIneq}). These new identities also enabled a refined error bound for the truncated prolate sampling formula that gives a rigorous grounding for the long-standing engineering folk theorem (Section~\ref{sec:prolate_sampling}).
Much of the mathematical depth in prolate Fourier theory lies in a commutation relation between an integral and a differential operator. In this work, we highlighted a spectral identity for the differential operator that mirrors Landau's classical threshold result for the integral operator (Section~\ref{sec:dual_signature}, \cite{landau_density_1993}). We provided a simplified proof of this inequality, which strengthens a previously known bound \cite{osipov_analysis_2013}.

For the study of molecular systems, we combined our classical signal processing protocol with the Hadamard test circuit to generate samples of the autocorrelation function.  
This new algorithm we call quantum prolate diagonalization.
Our numerical experiments show that QPD exhibits Heisenberg scaling, with the estimation error decreasing proportionally to $1/T_{\mathrm{runtime}}$ over several orders of magnitude in $T_{\mathrm{runtime}}$, for both single- and multiple-eigenvalue estimation.
Furthermore, our numerical experiments indicate that the average error of the eigenvalues decreases quadratically faster with the maximal evolution time compared to most recent results in the literature \cite{ding2024quantum}.
Our finding is similar to what is known in the literature for the standard QPE algorithm with a cosine filter applied on the readout register \cite{rendon2022effects}. 

Unlike most quantum subspace methods~\cite{lee2025efficient}, QPD is resilient to shot noise and does not suffer from conditioning issues, as demonstrated in our simulations. The theory of subspace protocols effectively resolves the conditioning issue through a trial space refinement and accurate dimension detection (Section A). This has been a long-standing issue for quantum subspace methods. Previous approaches addressed it with a thresholding strategy applied to the overlap matrix~\cite{klymko2022real, lee2024sampling, shen2025estimating}, but the number of measurements required to achieve chemical accuracy was on the order of $10000$ measurements per evaluated matrix element. In stark contrast, QPD achieves the same level of chemical accuracy with only $10$--$100$ shots per sample of the autocorrelation function.

We validated the performance of QPD on realistic quantum chemistry problems, with special attention paid to the role of the input state for the observed accuracy. These results served also as a test for QPD, because, as the recent work by Kiss et al.\,\cite{kiss2025early} highlights, practical implementations of early fault-tolerant algorithms can expose details that may be overlooked from a purely theoretical standpoint.

We have demonstrated that the concentration properties of prolates allow us to accurately estimate multiple eigenvalues with very short evolution times. 
As a future perspective to further reduce the computational cost of QPD we plan to implement sparse sampling approaches\,\cite{tang2013compressed,castaldo2025heisenberg} which may decrease the number of necessary signal samples. 

From a general perspective, not solely focusing on quantum computation, the results presented in this work mark progress toward a deeper understanding of optimal computation. Developing approximation routines to a given accuracy with minimal resources 
require a rigorous understanding of all contributing error sources. This, in turn, allows us to allocate computational resources across different subroutines to balance their respective error contributions. Here, we have taken a first step in this direction by providing a rigorous description of both subspace-based and discretization error.

\section*{Acknowledgments}

The authors acknowledge financial support from the Swiss National Science Foundation through Grant No. 200021\_219616 and from the Novo Nordisk Foundation (Grant No. NNF20OC0059939 'Quantum for Life'). This work has been presented as a contributed talk to the WATOC conference in Oslo in June 2025.

\begin{appendix}
\section{Subspace protocols}
\label{sec:SubspaceProtocols}

To facilitate the derivation of accuracy guarantees for such subspace-based methods, we established Theorem~\ref{thm:subspace_master} as a master theorem. In fact, the accuracy guarantees of PFD (Theorem \ref{thm:PFD_accuracy}) and sampled PFD (Theorem \ref{thm:sampledPFD}) follow as corollaries of Theorem~\ref{thm:subspace_master}. In this section, we briefly review subspace-based methods and present key results adapted from Ref.~\cite{stroschein2025approximationframeworksubspacebasedmethods}.

We assume a self-adjoint operator $H$ acting on a Hilbert space $\Hil$ and consider the eigenvalue problem 
\begin{align}
    H \varphi_i = \lambda_i \varphi_i
\end{align}
with $\varphi_i \in \Hil$ and $\lambda_i \in \mathbb{R}$. Subspace-based methods aim to approximate the eigenvalues of $H$ within a spectral subspace $\mathcal{E}[a,b]$ using a \emph{generalized eigenvalue problem} (GEP).
A numerical routine $\mathsf{P}$ generates a set of trial vectors $v_1 ,\cdots , v_m \in \Hil$ such that they \emph{approximately} span the spectral subspace $\mathcal{E}[a,b]$. Let $V\colon \mathbb{C}^m \to \Hil$ be the linear map defined by $ b \mapsto \sum_i v_i b_i$ and consider the subspace eigenvalue problem
\begin{align}
   V^\dagger H V b_i = \tilde \lambda_i V^\dagger V b_k. \label{eq:GEP_hidden}
\end{align}
Here, $b_k \in \mathbb{C}^m$ and $\tilde \lambda_i \in \mathbb{R}$. The generalized eigenvalues $\tilde \lambda_i$ serve as approximations to the eigenvalues of $H$ contained in the interval $[a,b]$.
Ref.~\cite{stroschein2025approximationframeworksubspacebasedmethods} develops a mathematical framework to quantify the accuracy of such subspace-based methods. To enable a rigorous analysis, the conceptual abstraction of \emph{subspace protocols} is introduced as a unifying formalism:

\begin{definition}[subspace protocols]
   % \label{def:epsilon_dim_red}
   \label{def:SubspaceProtocol}
    A numerical routine $\mathsf{P}$ is called a \emph{subspace protocol} if, for any self-adjoint operator
    $H$ on any Hilbert space $\Hil$ and any spectral subspace of $\mathcal{E}$ of $H$, the protocol $\mathsf{P}$ generates for any given trial dimension $M\in \mathbb{N}$ a generalized eigenvalue problem of the form
      \begin{align}
   A_M b_k = \tilde \lambda_k B_M b_k \label{eq:GEP_simple}.
      \end{align}
       Here  $A_M, B_M \in \mathbb{C}^{M \times M}$ are self-adjoint matrices, $b_k \in \mathbb{C}^M$ and $\tilde \lambda_k \in \mathbb{R}$. We denote the GEPs as \eqref{eq:GEP_simple} as $[A_M, B_M]$ and refer to $B_M$ as the \emph{weight matrix}.
    For a fixed pair $(H, \mathcal{E})$ we write $P_{(H, \mathcal{E})}(M)$ for the GEP generated by $\mathsf{P}$ at a trial dimension $M$.
       
Moreover, we assume the GEPs of a subspace protocol admit decompositions
      \begin{align}
   A_M & = V^\dagger_M H V_M + \delta A_M\\
   B_M & = V^\dagger_M V_M + \delta B_M,
      \end{align}
      where $V_M:\mathbb{C}^M \rightarrow \Hil $ is a liner map, called the \emph{trial vector map}, and $\delta A_M, \delta B_M \in \mathbb{C}^{M\times M}$ are self-adjoint matrices modeling numerical error. Let $P_\mathcal{E}$ denote the orthogonal projection operator onto $\mathcal{E}$ and define $N_M := P^\perp_{\mathcal{E}} V_M$. Then, the matrix
      \begin{align}
          \mathcal{N}^{(B)}_M  := N_M^\dagger N_M + \delta B_M
      \end{align}
      quantifies the deviation from the ideal projection onto $\mathcal{E}$ and is referred to as the \emph{noise weight}.
\end{definition}
The matrices $\delta A_M, \delta B_M$ can capture deviations in the computation of the matrix elements of 
 $V_M^\dagger H V_M$ and $V^\dagger_M V_M$. For example, in sampled PFD of Section \ref{sec:sampledPFD}, these matrix elements are approximated through the prolate sampling formula. Noise in processed data may introduce additional deviations, which can also be modeled by $(\delta A_M, \delta B_M)$.
It is advantageous to have analytical bounds on the noise weights --- particularly for a reliable detection of spectral subspace dimensions. This motivates a stricter class of subspace protocols:
\begin{definition}[$\epsilon$-subspace protocols]
\label{def:eps_SubspaceProtocol}
      Let $\mathsf{P}$ be a subspace protocol as in Definition \ref{def:SubspaceProtocol}, and let 
       $\{\epsilon_M\}_{M=1}^\infty$ be a sequence of non-negative real numbers that bound 
       the noise weights:
      \begin{align}
         \| \mathcal{N}^{(B)}_M \| \leq \epsilon_M \qquad \text{ for all } M \in \mathbb{N}  \label{eq:eps_M}
     \end{align}
     Then, the pair $(\mathsf{P},\{\epsilon_M\}_{M=1}^\infty )$ is called an \emph{$\epsilon$-subspace protocol} and we denote it by $\mathsf{P}^{\{\epsilon_M\}}$.
\end{definition}
The purpose of a subspace protocol is to approximate the eigenvalues of $H$ in the spectral subspace $\mathcal{E}$. However, the dimension of $\mathcal{E}$ is often unknown a priori. In such cases, the dimension can be inferred either from a significant drop in the spectrum of the weight matrix $B_M$ or by applying a noise threshold $\epsilon_{\text{th}} > 0$.

During the numerical operation, the protocol $\mathsf{P}$, the operator $H$, and the target subspace $\mathcal{E}$ remain fixed, while the trial dimension $M$ and the threshold $\epsilon_{\text{th}}$ are variable parameters. Accordingly, we define
\begin{align}
    \mathsf{P}_{(H, \mathcal{E})}(M, \epsilon_{\text{th}})
\end{align}
as an instance of Algorithm~\ref{alg:eps_dim_red}, specified by fixed choices of $\mathsf{P}$, $H$, and $\mathcal{E}$. The algorithm takes $(M, \epsilon_{\text{th}})$ as input, detects the dimension of $\mathcal{E}$, and returns approximations of the eigenvalues of $H$ in $\mathcal{E}$:

\begin{algorithm}[H]
\caption{Operation of a subspace protocol $\mathsf{P}_{(H, \mathcal{E})}(M, \epsilon_{\text{th}})$}
\label{alg:eps_dim_red}
\begin{enumerate}
    \item Use $\mathsf{P}$ to compute the matrices $A_M, B_M \in \mathbb{C}^{M \times M}$, where $M$ is the current trial dimension. Diagonalize $B_M$ and obtain its eigenvalues.
    \item Determine the \emph{detected dimensionality} $m$ as the number of eigenvalues of $B_M$ exceeding the threshold $\epsilon_{\text{th}}$. If $m = M$, increase the trial dimension $M$ and repeat the process.
    \item Let $U_m \in \mathbb{C}^{M \times m}$ be the matrix whose columns are the $m$ leading eigenvectors of $B_M$. Form the projected matrices $A_M^m := U_m^\dagger A_M U_m$ and $B_M^m := U_m^\dagger B_M U_m$. The resulting GEP $[A_M^m, B_M^m]$ is of dimension $m$ and satisfies $\lambda_m(B_M^m) > \epsilon_{\text{th}}$, ensuring it is well-conditioned.
    \item Compute the generalized eigenvalues of $[A_M^m, B_M^m]$. These serve as approximations to the eigenvalues of $H$ in the subspace $\mathcal{E}$.
\end{enumerate}
\end{algorithm}
If, in step~2, all eigenvalues of $B_M$ are smaller than $\epsilon_{\text{th}}$, a zero subspace is detected, and we set $m = 0$. This can occur either when the target spectral subspace $\mathcal{E}$ is trivial (i.e., $\mathcal{E} = {0}$), or when the weight matrix $B_M$ is too poor to distinguish signal from noise.
The refinement in step $3$ restricts the GEP to the dominant subspace of $B_M$. This projection improves numerical stability and allows to use an initial trial dimension $M$ that exceeds $\dim(\mathcal{E})$.

If $\mathsf{P}$ is an $\epsilon$-subspace protocol, the noise threshold $\epsilon_{\text{th}}$ is determined by the analytic bounds $\epsilon_M$ satisfying Eq.\ \eqref{eq:eps_M}. In particular, since the noise level generally increases with the number of trial vectors, an appropriate threshold $\epsilon_{\text{th}}$ is typically dimension-dependent:
\begin{align}
\epsilon_{th}(M) = \epsilon_M.
\end{align}
For $\epsilon$-subspace protocols, Theorem 1 of Ref.\cite{stroschein2025approximationframeworksubspacebasedmethods} guarantees that dimension detected by \\ $\mathsf{P}^{ \{ \epsilon_M \}}_{(H,\mathcal{E})}(M, \epsilon_M)$ does not exceed $\operatorname{dim} \mathcal{E}$.

Increasing $M$ beyond $\dim(\mathcal{E})$ enhances the expressiveness of the trial vector space, making it more likely to capture the signal subspace and to reveal a sharp spectral gap in $B_M$. This improves both the reliability of dimension detection and the conditioning of the refined 
GEP~\cite[Section~4.1]{stroschein2025approximationframeworksubspacebasedmethods}. 
However, this benefit comes at a cost: higher $M$ typically also increases the noise level $\epsilon_M$. Therefore, the choice of $M$ presents a trade-off between improved signal representation and increased contributions from noise.

If the dimension of $\mathcal{E}$ is known in advance or has been reliably detected using $\mathsf{P}_{(H,\mathcal{E})}(M, \epsilon_{\text{th}})$, we define a related protocol $\mathsf{P}_{(H,\mathcal{E})}(M, m)$, which takes the dimension $m$ directly as input instead of a noise threshold. The algorithm $\mathsf{P}_{(H,\mathcal{E})}(M, m)$ is identical to Algorithm~\ref{alg:eps_dim_red}, except that step~2 is omitted.

Once $\mathsf{P}_{(H,\mathcal{E})}(M, \epsilon_{\text{th}})$ has correctly identified the dimension $m$, Theorem~\ref{thm:subspace_master} provides an accuracy guarantee for the eigenvalues returned by $\mathsf{P}_{(H,\mathcal{E})}(M, m)$. The theorem distinguishes between two sources of error:
a \emph{subspace-based error}, arising from $ \operatorname{Range} V_M \neq \mathcal{E} $ and
arbitrary other errors, modeled by perturbations $\delta A_M$ and $\delta B_M$ in the matrix elements of the GEP.
The subspace-based error is quantified by the measure $ \varepsilon_M^m : \mathcal{B}(\mathbb{R}) \to \mathbb{R} $, defined as
\begin{align}
    \varepsilon (I) = \operatorname{Tr}[V^\dagger P_{\mathcal{E}^\perp} P^{(H)}(I) V],  \label{eq:error_meas}
\end{align}
where $P^{(H)} $ is the unique projection-valued measure associated with $ H $, such that $ H = \int \lambda \, dP^{(H)}(\lambda) $.
We denote by $\tilde \lambda_1 \geq \cdots \geq \tilde \lambda_m$ the generalized eigenvalues of the GEP $[A_M^m, B_M^m]$, and by $\lambda_1 \geq \cdots \geq \lambda_m$ the eigenvalues of $H$ restricted to $\mathcal{E}$, where $m = \dim \mathcal{E}$.
To compactly handle indefinite behavior in self-adjoint error matrices $E$ (such as $ \delta A_M^m -\tilde \lambda_i \delta B_M^m $) we adopt the notation
$\lambda_{\min}^*(E) := \min(\lambda_{\min}(E), 0)$ and $\lambda_{\max}^*(E) := \max(\lambda_{\max}(E), 0)$.

\begin{theorem}
    \label{thm:subspace_master}
    Let $\mathsf{P}$ be a subspace protocol as in Definition~\ref{def:SubspaceProtocol} and $m = \dim \mathcal{E}$. Consider a trial dimension $M \geq m$ and let $[A_M^m, B_M^m]$ be the GEP returned by $\mathsf{P}_{(H, \mathcal{E})}(M, m)$. Assume that the GEP is well-conditioned in the sense
    \begin{align}
        \lambda_m(B_M) > \| \mathcal{N}^{(B)}_M \|
    \end{align}
    and $\tilde \lambda_i \in [a,b]$ for all $i$. 
    Then, for each $i = 1, \dots, m$, the eigenvalues of $[A_M^m, B_M^m]$ satisfy
    \begin{align}
 \frac{\int_{-\infty}^a (\lambda - \tilde \lambda_i) d \varepsilon (\lambda) + \lambda_{\min}^*( \delta A_M^m -\tilde \lambda_i \delta B_M^m )}{\lambda_m(B_M) - \| \mathcal{N}^{(B)}_M \|}  \leq \tilde \lambda_i - \lambda_i \leq \frac{\int_b^\infty (\lambda - \tilde \lambda_i) d \varepsilon(\lambda) + \lambda_{\max}^*(\delta A_M^m -\tilde \lambda_i \delta B_M^m) }{\lambda_m(B_M) - \| \mathcal{N}^{(B)}_M \|}.
    \end{align} 
\end{theorem}

\section{Prolate spheroidal wave functions and essential dimensionality}
\label{sec:prolates}
In this section we provide a review of PFT that aims to promote the development of approximation routines with optimal accuracy guarantees. 
A core result of PFT is the $2WT$ theorem on the \emph{essential dimension} of the space of functions of bandwidth $W$ and approximate duration $2T$.
The $2WT$ theorem delineates an intriguing path to fundamental approximation guarantees rooted in spectral and dimensional optimality.

While our exposition builds on this theory, we make some adaptations and introduce extensions. We present a form of the $2WT$ theorem that highlights the eigenvalues of integral operator of time-and-band-limiting $\mathcal{B}_W \mathcal{D}_T$ as fundamental accuracy bounds and improves on a constant pre-factor with respect to the original theorem in Ref.\cite{ProIII}. The statement on the essential dimensionality of the $2WT$ theorem is tied to a sharp transition region the spectrum of 
$\mathcal{B}_W \mathcal{D}_T$. The transition region is precisely located through a sequence of inequalities proven by Landau in Ref.\cite{landau_density_1993}. 
We refer to such a sequence of inequalities as \emph{spectral signature} and note that the spectral signature of the integral operator has a dual in the spectrum of differential operator which commutes with $\mathcal{B}_W \mathcal{D}_T$.

We consider PFT to be an essential corner stone of approximation theory which can guide the formulation of optimal approximation methods. To facilitate the applicability of PFT we introduce novel concentration inequalities, on the derivatives of prolates as well as a bound
on their maximal magnitude outside of the concentration inequalities (Section \ref{sec:newIneq}, below).
 The new inequalities enable sharp accuracy guarantees for PFD (Section \ref{sec:PFD}) and a reformulation of the prolate sampling formula \cite{walter_sampling_2003} that highlights its relation to the $2WT$ theorem (Appendix \ref{sec:prolate_sampling}, below). 
 The proof also relies on the differential equation which commutes with the integral operator.

\subsection{Prolates as optimal filtering system and the \texorpdfstring{$2WT$}{2WT}\,Theorem}

We motivate the significance of prolates and their properties by sketching their derivation as optimal filter systems.
An error measure that integrates over the energy regions outside of $[-W,W]$ describes the 
accuracy of a subspace protocol, as seen in Theorem \ref{thm:subspace_master}. For filter diagonalization, this suggests that an optimal filtering system is a sequence of time-limited functions which is maximally confined inside the frequency band $B_W$. We are led to the following optimization problem
\begin{align}
    \sup_{f \in \mathscr{D}_T} \frac{ \| \mathcal{B}_{W} f\|^2 }{\| f\|^2 }= \sup_{f \in \mathscr{L}^2(\mathbb{R})} \frac{ \| \mathcal{B}_{W} \mathcal{D}_{T} f\|^2 }{\| \mathcal{D}_{T} f\|^2 } \label{eq:optimal_concen},
\end{align}
where $\mathcal{B}_{W}$ is the projection operator onto the space of band-limited functions $\mathscr{B}_W$ and 
$\mathcal{D}_{T}$ onto the space of time-limited functions $\mathscr{D}_T$. 
Functions in $\mathscr{B}_W$ have Fourier transforms with support in $[-W,W]$. 
Refs.\ \cite{ProI, ProII} show how optimization problems such as Eq.\ \eqref{eq:optimal_concen} lead to the discovery of prolate spheroidal wave functions as eigenfunctions to the integral operator $\mathcal{B}_{W} \mathcal{D}_{T}$,
\begin{align}
    \mathcal{B}_{W} \mathcal{D}_{T} f(\tau) & = \int_{-T}^T  \frac{\sin W (\tau - t )}{\pi (\tau -t)} f(t) dt.  \label{eq:PSWF_int}
\end{align}
The remarkable properties of prolates have been studied in a seminal 
series of papers \cite{ProI, ProII, ProIII, FUCHS1964317, SlepianAsymp, OnBandwidth,  landau_eigenvalue_1980,landau_density_1993}. 
Theorem \ref{thm:Prlt_simple_summary} summarizes a selection of key properties adopted from Ref.\cite{ProI}:

\begin{theorem}
\label{thm:Prlt_simple_summary}
    There exists a sequence of functions $\{\prlt_n \}_{n=0}^\infty$ in $\Ls_{\infty}^2$ with the following properties:
    \begin{enumerate}
        \item For all $t\in \mathbb{C}$ 
        \begin{equation}
            \mathcal{B}_{W} \mathcal{D}_{T} \prlt_n(t) = \gamma_n \prlt_n(t), \label{eq:PSWF_int_eig}
        \end{equation}
        where $1> \gamma_0 > \gamma_1 > \cdots $ and $\gamma_n$ converges to $0$ for $n\to \infty$.
        \item The system $\{\prlt_n \}_{n=0}^\infty$ is complete and orthogonal in $\mathscr{B}_{W}$ and $\Ls^2_{[-T,T]}$.
        \item The energy of $\prlt_n$ in the time interval $[-T,T]$ is given by $\gamma_n$. In particular, 
         \begin{align}
        \int_{-\infty}^{\infty} \prlt_n(t) \prlt_m(t) dt = \delta_{nm} \qquad \text{and}  \qquad \int_{-T}^{T} \prlt_n(t) \prlt_m(t) dt = \gamma_n \delta_{nm}.  \label{eq:double_ortho}
    \end{align} 
    \item The functions $\prlt_n(t)$ are real-valued for $t \in \mathbb{R}$, even for $n$ even and odd for $n$ odd. 
    \item For all $t \in \mathbb{C}$
    \begin{align}
        \int_{-T}^{T} e^{i \frac{ \bw \tau t}{T} } \prlt_n(t) dt  = \mu_n \sqrt{\frac{T}{\bw}} \prlt_n(\tau),  \label{eq:FiniteFT_op}
    \end{align}
    where $\mu_n\in \mathbb{C}$ and $\frac{|\mu_n|^2}{2\pi} = \gamma_n$.
    \end{enumerate}
\end{theorem}
The eigenvalues of $\mathcal{B}_{W} \mathcal{D}_{T}$ are functions 
of the product $WT$. To make this explicit we write $\gamma_n(c)$, where $c= WT$. The integral operator $\mathcal{B}_{W} \mathcal{D}_{T}$ is Hilbert-Schmidt in $\Ls_\infty^2$, and the supremum in \eqref{eq:optimal_concen} is attained by $\mathcal{D}_T \prlt_0$. 

The optimal concentration of prolates as well as their peculiar double completeness relation enabled the deep result that is the $2 \bw T $ theorem \cite{SlepianComment}. The theorem states that the space of functions of \emph{essential} bandwidth $W$ and \emph{essential} duration $2T$, is approximately $2WT/ \pi$-dimensional. The $2WT$ theorem was first proved in Ref.\ \cite{ProIII}. 
First, we give the formal definition of the space of band-limited functions, which is up to an error $\epsilon_T^2$ of duration $2T$:

\noindent \textit{
\label{def:2WTfunctions}
 Let $E(\epsilon_{T}) \subset \Ls_\infty^2$ be the set of functions of total energy $1$, bandwidth $W$, and 
  energy outside of $[-T,T]$ at most  $\epsilon_T^2$, that is,
\begin{align}
    \int_{|t|> T} |f(t)|^2 dt \leq \epsilon_{T}^2
\end{align}
for all $f \in E(\epsilon_T)$.}
The extended $2WT$ theorem, as we write it here, encompasses that the space $E(\epsilon_T)$ is optimally approximated by the prolate sequence. To make this statement precise, the concept of deflection is used:

\noindent \textit{
\noindent Let $\{\psi_n\}$ be a sequence of functions in $\Ls^2_\infty$ and $N\in \mathbb{N}$. We call
\begin{align}
    \delta(\{\psi_n\}, N) := \sup_{f \in E(\epsilon_T)} \inf_{a_n} 
    \displaystyle\int_{-T}^{T} \left| f(t) - \textstyle\sum\limits_{n=0}^{N-1} 
    a_n\, \psi_n(t) \right|^2 dt
\end{align}
the deflection of $E(\epsilon_T)$ from $\{\psi_n\}$ at dimension $N$.
}
We are now in a position to state the $2WT$ theorem: 
\begin{theorem}[$2WT$ theorem]
\label{thm:2WT}
Prolates optimally approximate the space of band-limited and time-concentrated functions $E(\epsilon_T)$ in the sense that $\{\prlt_n\} $ minimizes the deflection 
$\delta( \cdot, N) $ for all $N$. In particular, 
\begin{equation}
    \sup_{f \in E(\epsilon_T)} \inf_{a_n} \displaystyle\int_{-T}^{T} \left| f(t) - \textstyle\sum\limits_{n=0}^{N-1} a_n\, \prlt_n(t) \right|^2 dt 
    < \frac{\epsilon_T^2}{1 - \gamma_N(c)}.
    \label{eq:2WT_bound}
\end{equation}
For $N -1 = \lfloor  2 W T / \pi  \rfloor $ we have $(1-\gamma_N(c))^{-1} \leq 2$. 
\end{theorem}
The original version of the $2WT$ theorem was stated as 
    \begin{align}
        \sup_{f \in E(\epsilon_T)} \inf_{a_n} \displaystyle \int_{-T}^{T} \left| f - \textstyle\sum\limits_{n=0}^{\lfloor 2 WT / \pi \rfloor } a_n \prlt_n  \right|^2 dt <  C \epsilon_T^2
    \end{align}
and it was shown that $C< 12$. We prefer Eq.\ \eqref{eq:2WT_bound} as it underlines the significance of $1-\gamma_n(c)$ in the description of the optimal relation between Fourier volume $c=W T$, dimensionality $n$, and accuracy. In particular, we consider $1-\gamma_n(c)$ as fundamental approximation parameters of broader significance in the derivation of more extensive approximation routines. 
Moreover, Eq.\ \eqref{eq:2WT_bound} and the spectral signature in Eq.\ \eqref{eq:gamma_signature} immediately allow us to improve the bound on the constant to $C<2$. The optimality of prolates in approximating $E(\epsilon_T)$ was established in Theorem $1$ of Ref.\,\cite{ProIII}.

Motivated by the $2WT$ theorem we call $\tilde c = 2WT/\pi $ the \emph{essential dimensionality} of $W$-band-limited and $T$-concentrated functions. The $2WT$ theorem is a consequence of the eigenvalue distribution of the integral operator $\mathcal{B}_W \mathcal{D}_T$, which is summarized in Theorem \ref{thm:prlt_spectrum}:

\begin{theorem}[Spectrum of $\mathcal{B}_{W} \mathcal{D}_{T}$]
\label{thm:prlt_spectrum}
The integral operator $\mathcal{B}_{W} \mathcal{D}_{T}$ is Hilbert-Schmidt in $\Ls^2_\infty$ and
    \begin{align}
        \sum_{n=0}^\infty \gamma_n(c) = \frac{2\bw T}{ \pi}.
    \end{align}
    The sequence of eigenvalues crosses a threshold around the essential dimensionality $\tilde c = 2 W T/ \pi$:
    \begin{align}
        \gamma_{\lfloor \tilde c -1 \rfloor}(c) \geq 1/2 \geq \gamma_{\lceil \tilde c \rceil}(c) \label{eq:gamma_signature}. 
    \end{align}
    There exists a constant $A$ independent of $c$, such that for all $\gamma \in (0,1)$ the number of eigenvalues $\gamma_n(c)$ in between $\gamma$ and $1-\gamma$ is bounded as
    \begin{align}
        \#\left\{k \mid \gamma \leq \gamma_k(c) \leq 1-\gamma\right\} \leq \frac{A}{\gamma(1-\gamma)} \log \frac{2\bw T}{ \pi}.\label{eq:transition_region}
    \end{align}
    The asymptotic expansion of $1-\gamma_n(c)$ for $n$ fixed and large $c$ is 
    \begin{align}
 1-\gamma_n(c) =  \frac{4 \sqrt{\pi} 2^{3 n} c^{n+\frac{1}{2}} e^{-2 c}}{n!} (1 +\mathcal{O}(c^{-1})). \label{eq:gammaexpansion}
\end{align}
\end{theorem}
Eq.\ \eqref{eq:gamma_signature} was established in Ref.\ \cite{landau_density_1993}. We refer to such a crossing in the eigenvalues as a \emph{spectral signature}. In Section \ref{sec:dual_signature} below, we 
point on a dual of Eq.\ \eqref{eq:gamma_signature} appearing in the spectrum of the prolate spheroidal wave equation. We adopted Eq.\ \eqref{eq:transition_region} from Refs.\,\cite{ProIII, landau_density_1993}. The asymptotic expansion of Eq.\ \eqref{eq:gammaexpansion} was derived in Refs.\,\cite{FUCHS1964317, SlepianAsymp}.

Theorem \ref{thm:prlt_spectrum} shows that the spectrum of $\mathcal{B}_{W} \mathcal{D}_{T}$ is highly clustered around two values: $1$ and $0$, with a logarithmically sharp transition region between.
Eq.\ \eqref{eq:gamma_signature} precisely locates the transition region, and it follows that the number of prolates, which are highly concentrated in the time interval, is approximately $\tilde c = 2WT/\pi $, up to a logarithmic correction.
Hence, it was the connection to the spectrum of $\mathcal{B}_W \mathcal{D}_T$ that allowed us to describe the essential dimension of the space of signals.
The significance of the eigenvalues $\gamma_n(c)$ as optimal approximation bounds is underlined by their definition through a Fourier space optimization problem, Eq.\ \eqref{eq:optimal_concen}.
In Theorem \ref{thm:2WT},  $1-\gamma_n(c)$ gives an optimal relation between
the duration-bandwidth product $c=WT$, the degrees of freedom used in approximation $n$, and the achieved accuracy.
Therefore, PFT delineates an intriguing path toward an approximation theory that
offers fundamental accuracy guarantees rooted in spectral and dimensional optimality. 
Inspired by the $2WT$ theorem, we aim to develop numerical routines with accuracy guarantees essentially
described by $1-\gamma_n(c)$.
For PFD, this endeavor requires additions to the set of known concentration properties on prolate functions, which we present in the next section, Section \ref{sec:newIneq}.

\subsection{New concentration identities}
\label{sec:newIneq}

To derive accuracy guarantees for larger information processing protocols using Prolate Fourier Theory, it was necessary to establish how the $\ell_2$-concentration of prolates extends to their derivatives, as well as to a supremum bound outside their concentration region. These new identities are presented in Theorem~\ref{thm:prlt_bound} and, in particular, yield inequalities of the form $\sim 1 - \gamma_n(c)$.
Consequently, the sharp spectral transition described in Theorem~\ref{thm:prlt_spectrum} carries over to these new identities. Thus, Theorem~\ref{thm:prlt_bound} provides additional tools for deriving fundamental approximation schemes that exhibit a sharp accuracy transition around a critical dimension $\tilde c = 2WT/\pi$.

However, the new concentration identities also rely on a commutation relation of the integral operator 
$\mathcal{B}_W \mathcal{D}_T$ with a second order differential operator. Specifically, prolates are also eigenfunctions to the \emph{generalized prolate spheroidal wave equation}:
\begin{align}
    (T^2 -t^2) \prlt_n''(t)&  - 2 t \prlt_n'(t) + \bw^2 (T^2 -t^2) \prlt_n(t)  = - \lambda_n \prlt_n(t), \label{eq:genPSWEq}
\end{align}
The eigenvalues of the differential operator are functions of $c= WT$ and have 
\begin{align*}
    -c^2 < \lambda_0(c) < \lambda_1(c) < \cdots < \cdots 
\end{align*}
were $\lambda_n(c) \to \infty$ for $n \to \infty$ and $c$ fixed.
The peculiar commutation relation has played a significant role in the discovery and development of Prolate Fourier Theory \cite{ProI, ProIV, ProV, OnBandwidth}. During the derivation of Theorem~\ref{thm:prlt_bound} we noticed, that the spectrum of the differential operator exhibits signature similar to the one of the integral operator eq.\eqref{eq:gamma_signature}. In Section \ref{sec:dual_signature} we further remark on the commutation relation and the more recently discovered dual property.

Theorem \ref{thm:prlt_bound} presents our new inequalities, which are proven in Ref.\cite[Chapter 2]{stroschein2024prolatespheroidalwavefunctions}: 
\begin{theorem}
    \label{thm:prlt_bound}
    The energy concentration of prolates transfers to their derivatives as 
    \begin{alignat}{4}
        &\| \prlt_n' \|_{>T}^2 && = (1- &&\gamma_n) && C_{\text{extra},n}^2  \label{eq:thm_derivConcen_extra}\\
        &\| \prlt_n' \|_{T}^2  && =  &&\gamma_n     && C_{\text{intra},n}^2 \label{eq:thm_derivConcen_intra}.
    \end{alignat}
    Prolates have bounds within and outside of the concentration region
    \begin{alignat}{3}
        \sup_{|t| \geq T}  \prlt_n(t)^2 &\leq (1- &&\gamma_n) &&C_{\text{extra},n}  \label{eq:thm_prop_bound_extra} \\
        \sup_{|t| \leq T}  \prlt_n(t)^2 &\leq &&\gamma_n && \tilde C_{\text{intra},n}  \label{eq:thm_prop_bound_intra},
    \end{alignat}
        Here,
    \begin{align}
        C_{\text{extra},n} &= \sqrt{\| \prlt_n' \|_{\infty}^2 - \frac{\lambda_n}{T}  \frac{\prlt_n(T)^2}{1-\gamma_n} }, \quad \text{ and } \quad C_{\text{intra},n} = \sqrt{\| \prlt_n' \|_{\infty}^2 +  \frac{\lambda_n }{T}  \frac{\prlt_n(T)^2 }{\gamma_n}},
    \end{align}
    and $ \tilde C_{\text{intra},n} = C_{\text{intra},n}  + \delta_{n0} \frac{\prlt_0(T)^2}{\gamma_0}$.
\end{theorem}
The prefactors satisfy the following inequalities:
\begin{itemize}
    \item For $n \leq  \lfloor 2WT/\pi \rfloor -1$  \begin{align} C_{\text{extra},n} \leq \left(C_n - \frac{\lambda_n}{2T}\right) \quad \text{ and } \quad C_{\text{intra},n} < \bw,  \label{eq:lambda_neg_C_estimate}
    \end{align}
    \item For $n \geq \lceil2WT/\pi \rceil$ \begin{align} C_{\text{extra},n} < \bw \quad \text{ and } \quad C_{\text{intra},n} \leq \left(C_n + \frac{\lambda_n}{2T}\right),
     \end{align}
\end{itemize}
where $C_n = \sqrt{\|\prlt_n'\|_{\infty}^2 + \frac{\lambda_n^2}{4T^2}}$ and we have the bound $\|\prlt_n'\|_\infty < \bw$.
$C_{\text{extra},n}$ is a function of $W$ and $T$. Sometimes we make this explicit by writing $C_{\text{extra},n} \equiv C_{\text{extra},n}(W,T)$. 
For $n \leq  \lfloor 2WT/\pi \rfloor -1$, we can estimate
\begin{align}
    C_{\text{extra},n} & < \bw \left( \sqrt{ 1 + \frac{c^2}{4}} + \frac{c}{2} \right)\\
    & = \bw \left( c +\frac{1}{c} + \mathcal{O} \left(\frac{1}{c^3} \right) \right).
\end{align}
In particular, $C_{\text{extra},n}$ grows much more slowly with $c$ compared to the exponential decay of $1-\gamma_n(c)$ described in  Eq.\ \eqref{eq:gammaexpansion}. Therefore, the bounds $(1-\gamma_n)C_{\text{extra},n}$  and $(1-\gamma_n)C_{\text{extra},n}^2 $ inherit from $1-\gamma_n$ the the characterization as an approximation parameter exhibiting a sharp transition as $n$ approaches $2WT/\pi$.

\subsection{A dual in the spectrum of an integral and a differential operator}
\label{sec:dual_signature}

A surprising commutation relation between an integral operator and a differential operator played a significant role in the development of PFT, its generalizations, and its numerical applications \cite{ProI, ProIV, ProV}. Slepian famously described this commutation relation as a "lucky accident" \cite{SlepianComment}, and it later inspired the development of bispectral theory \cite{Grunbaum2004ThePS, grunbaum2022matrixbispectralitynoncommutativealgebras, casper2021algebrascommutingdifferentialoperators, casper2024matrixvalueddiscretecontinuousfunctions}.
The bispectrality of the prolate spheroidal wave functions has indeed led to unexpected developments in recent years \cite{UVSpectrum}.
However, the geometry behind the commutation relation still appears mysterious. On this note, we point to a dualism in the spectrum of the integral operator and differential operator, which seems to have been unnoticed.

The spectral signature due to Landau is of conceptual depth for PFT:
\begin{align}
    \gamma_{\lfloor \tilde c \rfloor - 1}(c) \geq 1/2 \geq \gamma_{\lceil \tilde c \rceil}(c) 
\end{align}
Here, we notice that the spectrum of the prolate spheroidal integral operator exhibits a similar signature:
\begin{align}
    \lambda_{\lfloor \tilde c \rfloor -1} ( c) \leq -1 \quad \text{ and } \quad 0 \leq \lambda_{\lceil \tilde c \rceil} (c)   \label{eq:diff_signatur_and}
\end{align}
A similar statement has, in fact, already been proved in Ref. \cite[Theorem 3.1]{osipov_analysis_2013}:
\begin{align}
    \lambda_{\lfloor \tilde c \rfloor -1} ( c) \leq  0 \leq \lambda_{\lceil \tilde c \rceil} (c)  \label{eq:osipov}
\end{align}
However, the duality between this spectral signature and that of the integral operator appears to have remained unnoticed in Ref.~\cite{osipov_analysis_2013}.
The proof in Ref.\ \cite{osipov_analysis_2013} is based on the Pr\"ufer transformation and rather technical. We include here a simpler proof for the slightly stronger statement given in the lower inequality of Eq.\ \eqref{eq:diff_signatur_and}.

For the sake of simplicity, we set $T=1 ; W=c$ and consider
the angular part of the prolate spheroidal wave equation as a differential operator:
\begin{align}
    H_c = - \frac{d}{d x}\left[\left(1-x^2\right) \frac{d}{d x}\right]-c^2\left(1-x^2\right)
\end{align}
For $\varphi \in Q(H_c)$, where  $Q(H_c)$ denotes the form domain, we have
\begin{align}
    (\varphi, H_c \varphi) = \int_{-1}^1 \left(1-x^2\right) \left( |\varphi'(x)|^2 - c^2|\varphi(x)|^2 \right) dx.
\end{align}
The form domain of $H_c$ is given by $C^1[-1,1]$.
We prove the left inequality of Eq.\ \eqref{eq:diff_signatur_and} through the variational principle \cite[Theorem XIII.2, page 76]{reed_iv_1978}: 
\begin{lemma}
\label{lem:varPrin}
Let $H$ be a self-adjoint operator on a Hilbert space, bounded from below and with a purely discrete spectrum.
Assume its eigenvalues are indexed in increasing order:
\begin{align}
    \lambda_0 \leq \lambda_1 \leq \lambda_2 \leq \cdots 
\end{align}
Then, its $n$-th eigenvalue satisfies
    \begin{align*}
        \lambda_n = \inf_{S_{n+1} \subset  Q(H)} \sup_{\psi \in S_n } \frac{(\psi, H \psi)}{(\psi, \psi)}
    \end{align*}
    where $Q(H)$ is the form domain of $H$, and the infimum is taken over all $(n+1)$‑dimensional subspaces of $Q(H)$.
\end{lemma}

\begin{proof}[\textbf{Proof of $\lambda_{\lfloor \tilde c \rfloor -1}(c) \leq -1$}]
Assume $\lfloor  \frac{2 c}{\pi}\rfloor$ is odd. 
For $k$ odd we have
\begin{align*}
    \int_{k\pi/ 2c}^{(k+2)\pi / 2 c} (1-x^2)(\sin^2(cx) - \cos^2(cx)) dx = -\frac{\pi}{2 c^3}.
\end{align*}
Define 
\begin{align*}
    \varphi_k(x) &= 
    \begin{cases} 
        \sqrt{\frac{2c }{\pi}} \cos(c x) \quad & \text{for } \quad x \in \left[\frac{k\pi}{2c}, \frac{(k+2)\pi}{2c}\right], \\
        0 & \text{else}.
    \end{cases}
\end{align*}
The index $k$ runs through all odd integers. There are $2\lfloor  \frac{c}{\pi} -\frac{1}{2} \rfloor  +1$ such functions $\varphi_k$ with support contained in $[-1,1]$.
For  $\lfloor  \frac{2 c}{\pi}\rfloor$ odd holds $\lfloor  \frac{2 c}{\pi}\rfloor = 2\lfloor  \frac{c}{\pi} -\frac{1}{2} \rfloor  +1 $. 
By construction, $\{\varphi_k\}$ forms an orthonormal system in $\mathscr{L}^2[-1,1]$.
Let $\psi = \sum_{k} a_k \varphi_k$ with $\sum_k |a_k|^2 = 1$. Evidently $\psi \in C^1[-1,1] = Q(H_c)$.
In particular, 
\begin{align}
    (\Psi, H_c \Psi) &= \sum_k |a_k|^2 \frac{2 c^3}{\pi} \int_{k\pi/ 2c}^{(k+2)\pi / 2 c} (1-x^2)(\sin^2(cx) - \cos^2(cx)) dx = - 1.
\end{align}
By Lemma \ref{lem:varPrin} we have therefore shown that 
\begin{align}
    \lambda_{\lfloor  \frac{2 c}{\pi}\rfloor -1} (c) \leq -1 .   \label{eq:lower_bound}
\end{align}
For the case where $\lfloor  \frac{2 c}{\pi}\rfloor$ is even we have $\lfloor  \frac{2 c}{\pi}\rfloor = 2 \lfloor  \frac{ c}{\pi}\rfloor $. For $k$ even 
\begin{align}
    \int_{k\pi/ 2c}^{(k+2)\pi / 2 c} (1-x^2)(\cos^2(cx) - \sin^2(cx)) dx = -\frac{\pi}{2 c^3},
\end{align}
and we can construct analogously a subspace spanned by sine functions supported on mutually disjoint intervals. Hence, inequality \eqref{eq:lower_bound} holds for all $c > 0$.
\end{proof}

\section{Prolate sampling formula}
\label{sec:prolate_sampling}

The $2WT$ theorem is a milestone in mathematical electrical engineering, as it provides a rigorous foundation for the long-standing \emph{engineering folk theorem}, which originated in the early development of communication theory.
This folk theorem on the dimensionality of signals emerged from sampling principles such as the Whittaker–Shannon interpolation formula \cite{Nyquist1928, NoiseCommunication, SlepianComment}.

The Whittaker–Shannon formula reconstructs a signal of bandwidth $W$ from uniform samples taken at rate $W/\pi$. In particular, it assumes $2WT/\pi$ samples within $[-T,T]$. This has led to the intuition that $2WT/\pi$ independent measurements should suffice to accurately reconstruct a signal of bandwidth and approximate duration $2T$. However, the Whittaker-Shannon interpolation formula does not succeed in proving this statement, as it converges very slowly and also requires a lot of samples outside of the interval of interest.

Only the $2WT/\pi$ theorem ultimately established that approximately $2WT/\pi$ degrees of freedom suffice to accurately represent a signal.
However, strictly speaking, the $2WT$-Theorem does not fully reflect the original context of the engineering folk theorem, which is rooted in representations using equidistant samples. In practice, such equidistant sampling remains particularly relevant due to its simplicity and implementational convenience. 
We highlight that PFT also enables us to give a sampling-based formulation of the engineering folk theorem.
This result is formalized in Theorem~\ref{thm:SamplingEngFolk}:

\begin{theorem}
    \label{thm:SamplingEngFolk}
    Consider a band limited function $f \in \BL_\bw$ in the time domain $[-T,T]$. Let $\{\prlt_n\}$ be the sequence of $\bw T $-prolates 
    and $c= T\bw/\pi$. We denote the truncated sampling series $f_{N}(t)$ defined as,
    \begin{align}
        f_{N}(t) = \frac{\pi}{\bw} \sum_{n=0}^{N-1}  \sum_{k= -\lfloor T\bw/ \pi\rfloor }^{\lfloor T\bw/ \pi \rfloor} f\left(\frac{k \pi}{\bw}\right) \prlt_n\left(\frac{k \pi}{\bw}\right) \prlt_n(t). \label{eq:trunc_prlt_sampling}
    \end{align}
    The sampling series $f_{N}(t)$ approximates $f$ within the interval $[-T,T]$ with a precision guarantee of,
    \begin{align}
        \|f -f_{N}\|_{T}^2 \leq &   \left( \|f \|_{>T}^2+ 2 \frac{\pi}{\bw}\|f \|_{>T}\|f'\|_{>T} \right) \sum_{n=0}^{N-1} \gamma_n(1-\gamma_n) C_{n} \\
        & +  \frac{ \|f\|_{>T}^2 }{1 - \gamma_N} 
    \end{align}
    Here, $C_n = (1+2 \frac{\pi}{\bw}C_{\text{extra},n})$ with $C_{\text{extra},n}$ as in Theorem \ref{thm:prlt_bound}.
\end{theorem}
Theorem~\ref{thm:SamplingEngFolk} establishes that $2WT/\pi$ equidistant sample points suffice to reconstruct a signal of bandwidth $W$ with high accuracy, provided the signal is 
sufficiently concentrated within the sampled interval. A proof is given in Ref.~\cite[Chapter 3]{stroschein2024prolatespheroidalwavefunctions}.

Theorem~\ref{thm:SamplingEngFolk} builds upon the prolate sampling formula that was introduced in Ref.\ \cite{walter_sampling_2003}. However, the truncation estimate of Theorem \ref{thm:SamplingEngFolk} is sharper than in Ref.\cite{walter_sampling_2003}, as it accesses the bounds of Theorem \ref{thm:prlt_bound}.

To formulate a sampling result more closely aligned with the original engineering folk theorem, we also made a modification in the assumptions. Ref.\ \cite{walter_sampling_2003} as well as Ref.\ \cite[Chapter 6]{hogan_duration_2012} consider functions $f$ lying in the subspace spanned by leading prolates, $\operatorname{PSWF}_N = \{ \prlt_0, \cdots , \prlt_{N-1}\}$.
Theorem \ref{thm:SamplingEngFolk} focuses instead on the space of band-limited and time-concentrated functions
$E(\epsilon_T)$, which we regard as a more natural model for real-world signals. A direct application of the $2WT$ theorem, as stated in Theorem \ref{thm:2WT}, allows us to bound 
the error in the sampling error due to contributions outside of $\operatorname{PSWF}_N$.

\end{appendix}

%\bibliography{references}

\end{document}